%% file: main.tex
\documentclass{article}

     \usepackage[final]{neurips_2021}

\usepackage[utf8]{inputenc} 
\usepackage[T1]{fontenc}    
\usepackage{hyperref}       
\usepackage{url}            
\usepackage{booktabs}       
\usepackage{amsfonts}       
\usepackage{nicefrac}       
\usepackage{microtype}      
\usepackage{xcolor}         

\usepackage{times}
\usepackage[compact]{titlesec}
\usepackage{amsmath,amsthm,amssymb}
\usepackage[ruled,algo2e]{algorithm2e}
\usepackage{algorithm}
\usepackage{latexsym}
\usepackage{epic}
\usepackage{epsfig}
\usepackage{hyperref}
\usepackage{verbatim}
\usepackage[justification=centering]{caption}
\usepackage{enumitem}
\usepackage{color}
\usepackage{natbib}
\setcitestyle{square,numbers}
\usepackage{bbm}
\usepackage{caption}
\usepackage{subcaption}
\usepackage{setspace}
\usepackage{wrapfig}
\usepackage{thmtools,thm-restate}

\usepackage{float}
\newfloat{algorithm}{t}{lop}

\long\def\name{$\mathtt{PADD}$}

\title{The Limits of Optimal Pricing in the Dark\thanks{This work is supported by a Google Faculty Research Award and an NSF grant CCF-2132506.}}

\author{%
	Quinlan Dawkins\\
	Department of Computer Science\\
	University of Virginia\\
	\texttt{qed4wg@virginia.edu} \\
	\And
	Minbiao Han \\
	Department of Computer Science\\
	University of Virginia\\
	\texttt{mh2ye@virginia.edu} \\
	\AND
	Haifeng Xu \\
	Department of Computer Science\\
	University of Virginia\\
	\texttt{hx4ad@virginia.edu} \\
}

\input{macros}

\begin{document}

\maketitle

\begin{abstract}

A ubiquitous learning problem in today's digital market is, during repeated interactions between a \emph{seller} and a \emph{buyer}, how a seller can gradually learn  optimal pricing decisions based on  the buyer's past purchase responses. A fundamental challenge of   learning in such a strategic setup is that the buyer will naturally have incentives to manipulate his responses in order to  induce more favorable learning outcomes for him.  To understand the limits of the seller's learning when facing such a strategic and possibly manipulative buyer, we study a natural yet powerful buyer manipulation strategy. That is,  before the pricing game starts, the buyer simply commits to ``imitate'' a different value function by pretending to   always react optimally   according to this \emph{imitative value function}. 

We fully characterize  the optimal imitative value function that the buyer should imitate as well as the resultant seller revenue and buyer surplus under this optimal buyer manipulation.  Our characterizations reveal many useful insights about what happens at equilibrium.  For example, 
a seller with concave production cost will obtain essentially $0$ revenue  at equilibrium whereas the revenue for a seller with convex production cost is the  \emph{Bregman divergence} of her cost function between no production and certain production. 
Finally, and importantly, we show that  a more powerful class of pricing schemes does not necessarily increase, in fact, may be harmful to,  the seller's revenue.  
Our results not only lead to an effective prescriptive way for buyers to  manipulate  learning algorithms but also shed  lights on the limits of what a seller can really achieve when pricing in the dark. 

\end{abstract}

\input{intro}

\input{model}

\input{characterization}

\input{instantiate}
\input{concave}
\input{conclude}

\bibliographystyle{abbrv}
\bibliography{main}

\newpage
\appendix
\include{appendix}

\end{document}

%% file: macros.tex
\newtheorem{lemma}{Lemma}

\newtheorem{definition}{Definition}
\newtheorem{remark}{Remark}

\newtheorem{example}{Example}

\newtheorem*{conjecture*}{Conjecture}
\newtheoremstyle{nonindented}{1ex}{1ex}{}{}{\bfseries}{.}{.5em}{}
\newtheoremstyle{indented}{1ex}{1ex}{\itshape\addtolength{\leftskip}{0.6cm}\addtolength{\rightskip}{0.6cm}}{}{\bfseries}{.}{.5em}{}
\theoremstyle{nonindented}
\theoremstyle{indented}
\theoremstyle{plain}

\newcommand{\bvec}[1]{\boldsymbol{ #1 }}

\renewcommand{\hat}{\widehat}
\renewcommand{\tilde}{\widetilde}
\renewcommand{\bar}{\overline}

\def\min{\qopname\relax n{min}}
\def\max{\qopname\relax n{max}}
\def\argmin{\qopname\relax n{argmin}}
\def\argmax{\qopname\relax n{argmax}}

\def\Pr{\qopname\relax n{\mathbf{Pr}}}
\def\Ex{\qopname\relax n{\mathbf{E}}}

\newcommand{\RR}{\mathbb{R}}

\def\C{\mathcal{C}}

\def\P{\mathcal{P}}

\newcommand{\eat}[1]{}

\newcommand{\maxi}[1]{\mbox{maximize} & {#1 } & \\}

\newcommand{\st}{\mbox{subject to} }

\newcommand{\con}[1]{&#1 & \\}
\newcommand{\qcon}[2]{&#1, & \mbox{for } #2.  \\}
\newenvironment{lp}{\begin{equation}  \begin{array}{lll}}{\end{array}\end{equation} }
\newenvironment{lp*}{\begin{equation*}  \begin{array}{lll}}{\end{array}\end{equation*}}


%% file: intro.tex
\section{Introduction}
Pricing  is a basic question in microeconomics \cite{mas1995microeconomic} as well as a ubiquitous problem in today's digital markets \cite{den2015dynamic}.  In its textbook style setup, there are two agents: a  buyer (he)  and a seller (she). The seller produces $d$ types of \textit{divisible} goods for sale and has production cost $c(\bvec{x})$ for producing  \emph{bundle} $\bvec{x} \in \RR_+^d$ of the $d$ goods.  Using a standard linear pricing scheme with unit price  $p_i$ for goods $i\in[d]$, the seller charges the  buyer   $\bvec{x} \cdot \bvec{p}$ when he purchases bundle $\bvec{x} \in \RR_+^d$ under price vector $\bvec{p} \in \RR_+^d$.  Naturally,  the buyer's optimal bundle for purchase   depends on his \emph{value} function $v(\bvec{x})$ about the products. We assume that, given price vector $\bvec{p}$,   the rational buyer will always  pick the optimal bundle that maximizes his quasilinear utility $[v(\bvec{x}) - \bvec{x} \cdot \bvec{p}]$.  The key question of interest is how the seller can compute   the revenue-maximizing optimal prices assuming rational   buyer purchase behaviors.  


The above basic problem can be easily solved if the seller has access to the buyer's value function $v(\bvec{x})$. However, in practice,  this value function is usually private  and unknown to the seller. To address this challenge, there has been a rich line of recent research which looks to design the optimal pricing scheme against an unknown buyer. Some of these work  have adopted  learning-based approaches that aim to compute the profit maximizing price by interacting with the buyer and gleaning information about the buyer's utility function \cite{amin2015online,roth2016watch,roth2020multidimensional}.  These are motivated by the wide spread of e-commerce today where the seller can repeatedly interact with the same buyer or buyers from the same population with similar preferences. For example,  in \emph{online advertising},  ad exchange platforms learn to price advertisers from past behaviors; in \emph{online retailing},   retailers learn to price customers from their past purchase history;  and in \emph{crowdsourcing},  platforms learn to reward workers' efforts.  Another line of works look to design dynamic pricing mechanisms that dynamically adjusts the price based on the observed past buyer purchase behaviors with the objective of maximizing the aggregated total revenue   \cite{amin2013learning,amin2014repeated,mohri2014optimal,mohri2015revenue,vanunts2019optimal}. 

A key challenge of pricing against such unknown buyers, regardless through learning or dynamic pricing, is that the  buyer may manipulate his responses  in order to mislead the seller's algorithm and induce more favorable outcomes for themselves. This is a fundamental issue when learning from strategic data sources, and is particularly relevant in this pricing setup due to the  strategic nature of the problem.   In this case,  the buyer controls the information content to the seller,  therefore he naturally has incentives to utilize this advantage to gain more payoff by strategically misleading the seller's   algorithms.  Indeed, as observed in previous studies,  online advertisers may strategically respond to an ad exchange platform's prices in order to induce a lower future price \cite{korula2015optimizing}; Consumers strategically time their purchases in order to obtain lower prices at online retailing platforms  \cite{li2014consumers}.  

To understand the limits of the seller's optimal pricing against such a strategic and possibly manipulative buyer, we put forward a natural model which   augments the above basic pricing problem with only one additional step --- i.e.,  we assume that the game starts with the buyer  \emph{committing} to a different value functions, coined the \emph{imitative value function}, after which the seller will compute the optimal pricing scheme against this imitative buyer value function.  The buyer's commitment to an imitative value function captures a simple yet powerful buyer manipulation strategy that can be used against any seller learning algorithms ---  that is, the buyer simply  behaves consistently according to this  imitative value function during the entire interaction process. Intuitively, the motivation of such commitment assumption comes directly from the fact that the seller has no information about the buyer and  has to optimize pricing ``in the dark''. Consequently, if the buyer consistently behaves according to some different imitative value function, the seller is not able to distinguish this buyer from another buyer who truly has the value function (see more justifications of the commitment assumption in Section \ref{sec:model}). Moreover, such an imitation strategy is also easy to execute in practice by the buyer regardless whether the seller is learning from him or is adopting dynamic pricing schemes. In fact, the buyer could   even just  report his imitative value function to the buyer directly at the beginning of any interaction. In this situation, no learning or dynamic pricing will be needed as the buyer will indeed always behave according to the imitative value function. Therefore,  all the seller can do is  to directly apply the optimal pricing scheme for the buyer's imitative  value function.  

We remark that such imitation-based manipulation strategy has attracted much interest in   recent works, with similar motivations as us --- i.e., trying to understand how to manipulate learning algorithm or conversely, how to design strategy-aware learning algorithms to mitigate such manipulation.  However, most of these works have focused on the general Stackelberg game model \cite{gan2019imitative,birmpas2020optimally} as well as the Stackelberg security games \cite{gan2019manipulating,nguyen2019imitative,nguyen2020decoding}. Our optimal pricing problem is also a Stackelberg model and thus a natural fit for the study. The crucial difference between our work and previous studies is that both agents in our model have \emph{continuous} utility functions whereas all these previous works \cite{gan2019imitative,gan2019manipulating,nguyen2019imitative,birmpas2020optimally}  have discrete agent utility functions. Therefore, our model leads to a \emph{functional} analysis and optimization problem, which is more involved. Fortunately, we show that the optimal functional solution can still be characterized by leveraging the structure of the pricing problem.

\subsection{Our Results and Implications}\label{sec:results}
Given the effectiveness  and easy applicability  of the buyer's imitative strategy described above, this paper studies what the optimal buyer imitative value function is and how it would affect buyer's surplus and the seller's revenue.  Our main result provides a full characterization  about the optimal buyer imitative value function.  We show that the optimal buyer imitative value function $u^*$ features a specific  bundle of products $\bvec{x}^*$ that is most desirable to the buyer. Interestingly, it turns out that   $u^*$ is a Leontief-type  piece-wise linear concave value function \cite{allen1967macro}   such that the buyer would only  proportionally value a fraction of the desired bundle   $\bvec{x}^*$ and nothing else. 
Moreover, we also characterize the  seller revenue and buyer surplus at equilibrium as well as necessary and sufficient conditions under which the seller obtains strictly positive revenue.

The optimal buyer imitative value function turns out to depend crucially on the seller's production cost function. When the cost function is concave, we show that the optimal buyer imitative value function $u^*$  will ``squeeze'' the  seller revenue essentially to $0$. The fortunate news for the seller, however, is that finding out the $u^*$ turns out to be an NP-hard task for the buyer. In fact, we prove that it is NP-hard to find the $u^*$ that can guarantee a polynomial faction of the optimal buyer surplus. 
For the   widely  adopted \emph{convex}  production cost, we show that the equilibrium seller revenue is  the Bregman divergence  between production $\bvec{0}$ and the $\bvec{x}^*$ bundle mentioned above. This illustrates an interesting message that convex production  costs are ``better'' at  handling buyer's strategic manipulations. Note that production costs are indeed more often believed to be convex since  economic models typically assume that marginal costs increase as quantity goes up \cite{shephard1974law,samuelson1998microeconomics}. 

All our characterizations so far focus on the standard \emph{linear} pricing scheme. Our last result examines the possibility of using a more general class of pricing  schemes  to address the buyer's strategic manipulation. Surprisingly, we show that for the strictly more general class of concave pricing functions (i.e., the seller is allowed to use any concave function as a pricing function), the equilibrium will remain the same as described above. Therefore, the more general pricing schemes do not necessarily help to address buyer's manipulation behavior.  In fact, we show that there exist  examples where strictly \emph{broader} class of pricing functions leads to strictly worse seller revenue. This is because more general pricing schemes may ``overfit'' buyer's incentives, which renders it easier to manipulate. This illustrates an interesting phenomenon in learning from strategic data sources and shares similar spirit to  \emph{overfitting}  in standard machine learning tasks.  

\vspace{-0.03em}

\subsection{Additional Related Works}\label{sec:related}
Due to space limit, here we only briefly discuss the most related works while refer readers to Appendix \ref{append:sec:related} for more detailed discussions and comparisons. Closely related to ours is a recent study by  Tang and Zeng \cite{tang2018price}. They study  the bidders' problem of committing to ``imitate'' a fake value \emph{distribution} in auctions and acting consistently as if the  bidder's  value were from the fake  distribution. This is similar in spirit to our buyer's commitment to an imitative value \emph{function}. However, there is significant difference betweens  our setting and  that of \cite{tang2018price}, which leads to very different conclusions as well. Specifically, the seller in \cite{tang2018price} auctions a single indivisible item with no production costs whereas our seller sells multiple divisible items with production costs.   Another recent work \cite{nedelec2020robust} also  studies buyer's strategic manipulation  against seller's pricing algorithms, but in a single-item multi-buyer situation. Moreover, they assume a fixed  seller  learning strategy (thus not adaptive to buyer's strategy) with separated  exploration-exploitation phases motivated by \cite{amin2013learning}. 
Another  very relevant literature is  learning the optimal prices or optimizing aggregated total revenue by repeatedly  interacting with a single buyer   \cite{amin2013learning,amin2014repeated,mohri2014optimal,mohri2015revenue,vanunts2019optimal}.  These   works all focus on designing learning algorithms that can learn from strategically  buyer responses. Our work  complements  this literature by studying the limits of what learning algorithms can achieve.  Moreover, the setups of these previous works are also different from us -- they either assume buyer values are drawn from distributions  \cite{amin2013learning,amin2014repeated,mohri2015revenue}  or the seller sells a single indivisible good  with discrete agent utilities. Thus, their results are not comparable to us.

There have also been studies on learning the optimal prices from \emph{truthful} revealed preferences, i.e., assuming the buyer will honestly best respond to seller prices \cite{beigman2006learning,zadimoghaddam2012efficiently,balcan2014learning,pmlr-v119-zhang20f,amin2015online,roth2016watch,roth2020multidimensional}. Our works try to understand if the assumption of truthful revealed preferences does not hold and if the buyer will strategically respond to the seller's learning, what learning outcome could be expected when the buyer simply imitates a different value function that is  optimally chosen. From this perspective, these works serve as a key motivation for the present paper.   More generally, our work subscribes to the general line of research on learning from strategic data sources. 
Most works in this space has focused on classification  \cite{bruckner2011stackelberg,hardt2016stratclass,zhangincentive,dong:ec18,milli:fat19,hu:fat19,chen2020learning},  regression problems \cite{Perote2004StrategyproofEF,dekel2010incentive,chen:ec18} and distinguishing distributions \cite{zhang2019distinguishing,zhang2019samples}. Our work however focuses on learning the optimal pricing scheme.

%% file: model.tex
\section{Preliminaries}


\paragraph{Basic Setup of the Optimal Pricing Problem.} A seller (\emph{she}) would like to  sell $d$ different types of \emph{divisible} goods to a buyer (\emph{he}).  It costs her $c(\bvec{x})$ to produce $\bvec{x} \in \RR^d$ bundle of these goods. Let $X \subset \RR^d$ denote the set of all  feasible bundles  that the seller can produce.  We assume $X$ is convex, closed and has positive measure. As a standard assumption  \cite{mas1995microeconomic,roth2016watch,beigman2006learning,balcan2014learning,zadimoghaddam2012efficiently}, the buyer has a \emph{concave} value function $v(\bvec{x})$  for any goods bundle $\bvec{x} \in X$. We do not make any assumption about the seller's production cost   $c(\bvec{x})$, except that it  is monotone non-decreasing.  For normalization, we assume $\bvec{0} \in X$ and $v(\bvec{0} ) = c(\bvec{0} ) =0$.  

The seller aims to find a revenue-maximizing  pricing scheme  assuming rational buyer behaviors. A seller pricing scheme is a function $p(\bvec{x})$ that specifies the sale price for any bundle $\bvec{x}$. By convention,  $p(\bvec{0}) = 0$ always.  Let the set $\P$ denote the set of all pricing \emph{functions} that are allowed to use by the buyer.  The majority of this paper will  focus on the textbook-style \emph{linear} pricing scheme  \cite{mas1995microeconomic}. A linear pricing scheme is parameterized by a price vector $\bvec{p} $ (to be designed) such that $p (\bvec{x})= \bvec{p} \cdot \bvec{x}$ where $i$'th entry $p_i$   is interpreted as the unit price for goods $i$.   Let set \[\P_L = \{p: p(\bvec{x}) = \bvec{p} \cdot \bvec{x} \text{ for some } \bvec{p}   \in \mathbb{R}^d_{+} \}\]  denote the set of all linear pricing functions.  In Section \ref{sec:nonlinear-pricing} we will  also study the  broader classes of \emph{concave} pricing schemes where the set $\P$  consists  of  all  monotone non-decreasing concave  functions. 

 For any price function   $p \in \P$, a rational buyer looks to purchase bundle $\bvec{x}^*$ that maximizes his \emph{utility}; That is, $\bvec{x}^* = \arg \max_{\bvec{x} \in X} \big[ v(\bvec{x}) - p( \bvec{x})  \big]$. Ties are  broken in favor of the seller.\footnote{This is usually without loss of generality since the seller can always induce desirable tie breaking by providing a negligible additional incentive  to the buyer.}
The buyer's best response can thus be viewed as a function  $\bvec{x}^*(p)$ of the seller's price function $p(\bvec{x})$. Knowing that the buyer will best respond, the seller would like to pick the pricing function  $p^* \in \P$ to maximize her revenue. Formally, $p^*$ is the solution to the following bi-level optimization:
\begin{equation}\label{eq:opt-price}
	p^* = \argmax_{p \in \P}  [  p( \bvec{x}^*(p))  - c(\bvec{x}^*(p)) ],   \quad \text{ where }\,  \bvec{x}^*(p) = \argmax_{\bvec{x} \in X} \big[ v(\bvec{x}) - p( \bvec{x} ) \big] 
\end{equation}
The optimal solution $\big(p^*, \bvec{x}^*(p^*)\big)$ to such a bi-level optimization problem forms an \emph{equilibrium} of this pricing game.  More formally, this is often called the \emph{optimal}  Stackelberg equilibrium or  \emph{strong} Stackelberg equilibrium  \cite{conitzer2006computing,roth2016watch}. We   call $p^*$ the \emph{equilibrium pricing function} and $\bvec{x}^*$ the \emph{equilibrium bundle}.  Note that this is a challenging \emph{functional} optimization problem since the seller is picking a function $p \in \P$, while not a vector variable. However, when $\P = \P_L$ is the set of linear pricing scheme, the above problem becomes a bi-level variable  optimization problem since any $p(\bvec{x}) \in \P_L$ can be fully characterized by a price vector $\bvec{p}$. 


{\bf Terminologies from Convex Analysis. } We defer basic definitions like convex/concave functions and super/sub-gradients to Appendix \ref{app_sec:convex_background}, and only mention two useful notations here: (1) the set of super/sub-gradient for concave/convex function $f$ is denoted as $\partial f(\mathbf{x})$;  (2)  The  \emph{Bregman divergence} of a function $f$ is defined as   $D_{f}(\mathbf{z},\mathbf{x}) = f(\mathbf{z}) - f(\mathbf{x}) - \nabla f(\mathbf{x})\cdot [  \mathbf{z}- \mathbf{x}]$. $D_{f}(\mathbf{z},\mathbf{x})$ is an important distance notion and is strictly positive for strictly convex functions when  $\mathbf{z} \not = \mathbf{x}$.

\section{A Model of \texttt{P}ricing \texttt{A}gainst a \texttt{D}eceptive  Buyer in the \texttt{D}ark (\name) }\label{sec:model}


As mentioned in related work section \ref{sec:related}, the literature of  algorithms to learn pricing schemes from unknown buyers is massive. This work, however, takes a different perspective and seeks to understand how a buyer can strategically deceive the seller, through a simple yet effective class of manipulation strategies.  
Our model naturally captures a buyer's strategic responses to seller's pricing algorithms  when the seller has no prior knowledge about the buyer, i.e.,    pricing ``in the dark''.  

Thus, we study a buyer manipulation strategy that is  \emph{oblivious} to   any  pricing algorithm. That is, the buyer simply imitates a different   value function $u(\bvec{x})$ by consistently responding to \emph{any} seller acts  according to $u(\bvec{x})$. Consequently, whatever the seller learns will be with respect to this \emph{imitative value function} $u(\bvec{x})$.  Alternatively, one can think of the buyer as  \emph{committing} to always behave   according to   value function $u(\bvec{x})$. Nevertheless, the buyer's objective is still to maximize his \emph{true} utility by carefully crafting an imitative  value function $u(\bvec{x})$ to commit to. Given   the buye's commitment to  value function $u(\bvec{x})$, the best pricing scheme for  the seller is to use the optimal pricing function against buyer value function $u(\bvec{x})$.  Similar to $v(\bvec{x})$, the imitative value functions $u(\bvec{x})$ is  assumed to be concave and monotone non-decreasing as well. Let set $\C$ denote the set of all such functions. These  resulted in the following model of  \emph{Pricing Against a Deceptive buyer in the Dark} (\name{}): 
	\begin{itemize} 
		\item  The buyer with true value function $v(\bvec{x})$ (unknown to seller) commits to react optimally according to an \emph{imitative value function} $u(\bvec{x}) \in 		\C$. 
		\item  The seller learns the buyer's imitative value function $u(\bvec{x})$ and compute the optimal pricing function $p^* \in \P$ by solving bi-level Optimization Problem \eqref{eq:opt-price} w.r.t. $u(\bvec{x})$. 
		\item The buyer observes the seller's pricing function $p^*$, and then follows his commitment to react optimally w.r.t. to $u(\bvec{x})$ by purchasing bundle $\bvec{x}^* = \arg\max_{\bvec{x} \in X} [ u(\bvec{x}) - p( \bvec{x})]$. 
	\end{itemize}
 We remark that such  commitment to a fake value function is not uncommon  in the literature; similar assumptions have been adopted in many recent works in, e.g., auctions \cite{tang2018price}, general Stackelberg games \cite{gan2019imitative,birmpas2020optimally} and security games \cite{gan2019manipulating,nguyen2019imitative,nguyen2020decoding}.  The buyer's ability of making such a commitment  fundamentally comes from the fact that the seller has no prior knowledge about the buyer's true value function $v(\bvec{x})$, i.e., has to ``price in the dark''. We refer curious reader to  Appendix \ref{app_sec:commitment} for a more  detailed discussion about  this assumption. 
 
 Naturally,   the buyer  with true value function $v(\bvec{x})$  would like to find the optimal imitative value function $u^*(\bvec{x})$  to maximizes his utility. This results in the following  equilibrium definition. 
\begin{definition}[Equilibrium of \name]\label{def:equ}
The equilibrium of \name{} consists of the optimal imitative value \emph{function} $u^* \in \C $ that the buyer  commits to, the seller's optimal pricing function $p^* \in \P$  against $u^*(\bvec{x})$, and the buyer's response bundle $\bvec{x}^*\in X$. Formally,  $(u^*, p^*, \bvec{x}^*)$ is an equilibrium for a buyer with true value function $v(\bvec{x})$ to \name{} if 
\begin{eqnarray}\nonumber 
 u^* = \argmax_{u  \in \C}  [  v(\bvec{x}^*) - p^*(\bvec{x}^* ) ], \qquad   
 \text{ where } &&   p^* = \argmax_{p \in \P}  [  p( \bvec{x}^*)  - c(\bvec{x}^*) ],   \\ \label{eq:equilibrium-def} 
 \text{ where }  &&  \bvec{x}^* = \argmax_{\bvec{x} \in X} \big[ u^*(\bvec{x}) - p^*( \bvec{x} ) \big]. 
\end{eqnarray} 
 \vspace{-4mm}
\end{definition}
The equilibrium of \name{} gives rise to  a challenging \emph{tri-level functional optimization problem}.\footnote{Even for linear pricing where $\P = \P_L$, this is still a functional optimization problem since the buyer's imitative value function $u$ is an arbitrary function in $\C$. } Note that, the dependence of the buyer's objective $[  v(\bvec{x}^*) - p^*(\bvec{x}^* ) ]$ on $u$ is indirectly through the two argmax problems afterwards.  The buyer is assumed to  know the production cost function $c(\bvec{x})$, which is needed to compute his optimal $u^*$. This can be easily justified in situations where the seller has been on the market for some time, therefore her production cost gradually becomes public knowledge.

%% file: characterization.tex
\section{The Equilibrium of \name{} under Linear Pricing, and    Implications}\label{sec:characterization}
In this section, we characterize  the  equilibrium of    \name{}  under  linear pricing schemes, i.e., $\P = \P_L=\{p: p(\bvec{x}) = \bvec{p} \cdot \bvec{x} \text{ for some } \bvec{p}\in \mathbb{R}^d_{+} \}$  consists of all linear pricing functions. A linear pricing function is determined by a non-negative price vector   $\bvec{p}\in \mathbb{R}^d_{+}$. To distinguish \emph{functionals} from \emph{vector variables},  we will  use $P_L= \{  \bvec{p}:  \bvec{p}\in \mathbb{R}^d_{+}\}$ to denote the set of all  possible non-negative price vectors so that each  $\bvec{p} \in P_L$ uniquely corresponds to a linear pricing function in $\P_L$.   Under linear pricing, the equilibrium in Def. \ref{def:equ}  denoted by $(u^*, \bvec{p}^*, \bvec{x}^*)$ is characterized by Eq. \eqref{eq:equilibrium-def} where $\P = \P_L$.
 


Even with linear pricing,  this is still a very challenging \emph{tri-level functional} optimization problem since $u^*$ is a function  chosen from the set of all possible functions from $X$ to $\RR_+$, denoted by set $\C$.  A first thought one might have is: why doesn't the buyer simply imitate $u^*(\bvec{x}) = c(\bvec{x})$. We will see later that this is  not ---  in fact far from being --- optimal since it   makes the game zero-sum and the seller will pick a price that guarantee 0 revenue, e.g., a price of $\infty$. This leads the trade to happen at 0 production, which  is clearly not optimal for the buyer. The optimal $u^*$ should provide some incentive for the seller to produce some amount $\bvec{x}^*$ that is  in some sense the best for the buyer. We will provide two concrete examples in the next section in Figure \ref{fig:one-dimensional-characterization2}.   

The main result of this section is the following  characterization for the equilibrium of \name{}. This general characterization does not depend on any specific property about function $v(\bvec{x}), c(\bvec{x})$.  

\begin{restatable}{restatethm}{primethmcharacterization}
\label{thm:characterization}In the  equilibrium of \name{} under linear pricing, the optimal buyer imitative  value function $u^*$  can w.l.o.g. be written as the following concave function parameterized by  production amount $ \bvec{x}^*\in X$ and a real value $p^* \in \RR_+$: 
	\begin{equation} \label{eq:optimal-format}
	u^*(\bvec{x}) =p^*  \cdot \min \{ \frac{x_1}{x^*_1}, \cdots \frac{x_d}{x^*_d}, 1\}   
	\end{equation}
	where 
	\begin{equation}\label{eq:optimal-xp}
\bvec{x}^* = \arg \max_{\bvec{x} \in X} \big[  v(\bvec{x}) - \sup_{\alpha \in [0,1)} \frac{  c( \bvec{x}) -c(\alpha \bvec{x}) }{1 - \alpha} \big] \quad \text{ and } \quad p^* =  \sup_{\alpha \in [0,1)} \frac{  c( \bvec{x}^*) -c(\alpha \bvec{x}^*) }{1 - \alpha} 
	\end{equation}
	
	Moreover, under imitative value function $u^*(\bvec{x})$,  
	\begin{enumerate}
		\item For any   vector $\lambda \in \Delta_d$ in the $d$-dimensional simplex,  the linear pricing scheme with  price vector   $\bvec{p}^* =  (\lambda_1  \frac{p^*}{x_1^*}, \lambda_2   \frac{p^*}{x_2^*}, \cdots, \lambda_d   \frac{p^*}{x_d^*})$ is optimal for $u^*$. 
		\item In any of the above optimal linear pricing schemes, the buyer's optimal  bundle response is always $\bvec{x}^*$ and the buyer payment will always equal $p^*$. 
		\item At equilibrium, the buyer surplus  is $[v(\bvec{x}^*) - p^*]$ and the seller revenue is $[p^* - c(\bvec{x}^*)]$.   
	\end{enumerate}
\end{restatable}

Note that the $u^*$ described by Equation \eqref{eq:optimal-format} may not be the unique optimal buyer imitative value function, but it is one of the optimal ones. Moreover, any optimal imitative value function will result in the same buyer surplus and seller revenue. 
%


{\bf Interpretation of Theorem \ref{thm:characterization}. }  Before a formal proof,  it is worthwhile to take a closer look at the characterization of the buyer's optimal imitative value function $u^*$ characterized by Theorem \ref{thm:characterization} and the special pricing problem it ultimately induces.   At a high level, the buyer has a desirable amount of products $\bvec{x}^*$ in mind, characterized by the first equation in \eqref{eq:optimal-xp}. He would pretend that his value for $\bvec{x}^*$ equals $p^*$. Moreover, the value of any production amount $\bvec{x}$ will linearly depend on the  maximum  possible coordinate-wise \emph{fraction} of $\bvec{x}^*$ that $\bvec{x}$ contains, i.e.,  the $ \min \{ \frac{x_1}{x^*_1}, \cdots \frac{x_d}{x^*_d}, 1\} $ term in $u^*$.\footnote{Notably, the format of this  value function appears to have interesting connections to the well-known   \emph{Leontief   production function}  \cite{allen1967macro} which has the format of $\min_i \{ \frac{x_i}{a_i}\}$. However, leontief functions  are used  to describe seller's production \emph{quantities} as a function of quantities of different factors with no substitutability. It is an expected surprise that similar type of \emph{value} function turns out to be   optimal for buyer's strategic manipulation.  } 
 \begin{wrapfigure}{R}{0.25\textwidth}
	\centering
	\includegraphics[width=0.25\textwidth]{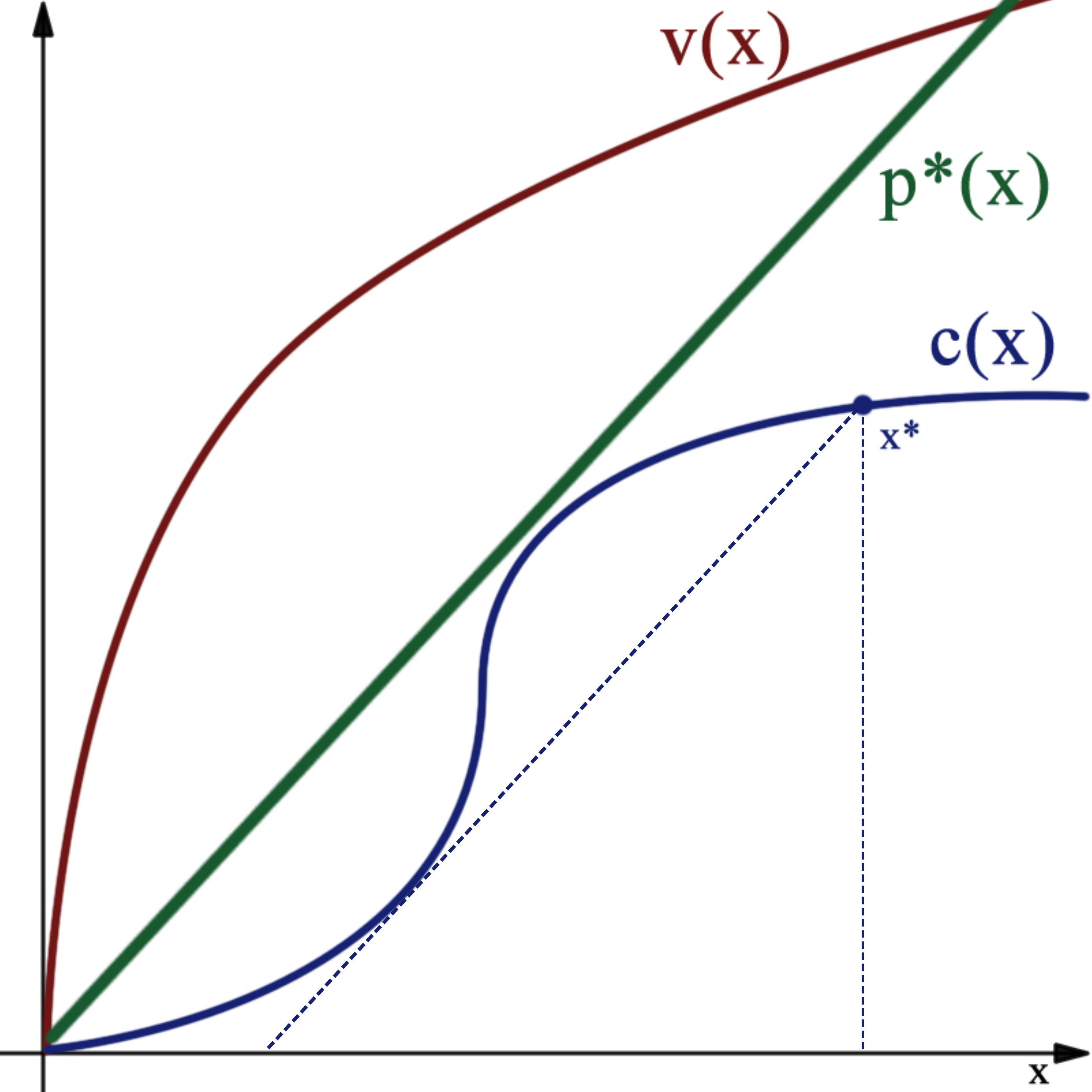}
	\caption{Illustration of $p^*$ in Theorem \ref{thm:characterization}.  }
	\label{fig:c_and_u}
\end{wrapfigure}
\
\

The $p^*$ is chosen as the largest possible \emph{slope} of the segment between $c( \bvec{x}^*)$ to $c( \alpha \bvec{x}^*)$ among all possible $\alpha \in [0,1)$ (see Figure \ref{fig:c_and_u} for an illustration in one dimension about how $p^*$ is chosen based on the cost function $c$). This is certainly a very carefully chosen value.  Note that  \[p^* =   \sup_{\alpha \in [0,1)} \frac{  c( \bvec{x}^*) -c(\alpha \bvec{x}^*) }{1 - \alpha}  \geq     \frac{  c( \bvec{x}^*) -c( \bvec{0}) }{1 - 0}  = c(\bvec{x}^*),\] which implies  non-negativity of  the seller's revenue $[p^* - c(\bvec{x}^*)]$. In fact, the seller can achieve \emph{strictly positive}  revenue if and only if the above $\geq$ is strict.  Finally, the  optimal buyer imitative value function $u^*$ leads to a special pricing problem for the seller.  Since the buyer's value is designed such that he is only interested in purchasing some fraction $r\in [0,1]$  of $\bvec{x}^*$, the seller's optimal pricing will    charge  $r p^*$ for the bundle $\bvec{x} = r \bvec{x}^*$  and this total charge $r p^*$ can be distributed \emph{arbitrarily} over the $d$ products, e.g., by charging  $\lambda_i  r p^* = \lambda_i \frac{p^*\cdot x_i}{x_i^*} $  for product $i$ where $\sum_{i=1}^d \lambda_i = 1$.  Overall, Theorem \ref{thm:characterization} fully characterized what the pricing problem is like at the equilibrium of \name{}.  
 
\begin{remark}\label{remark1}
Theorem \ref{thm:characterization} only provides a structural characterization about the equilibrium but does not imply that the equilibrium $(u^*, \bvec{p}^*, \bvec{x}^*)$ can be computed efficiently since we still need to solve the optimization problem  \eqref{eq:optimal-xp} to find the optimal $\bvec{x}^*$. As we will show later, this turns out actually to be an NP-hard problem in general. This is an interesting situation where the result reveals useful structural insights   despite  its computational intractability. 
\end{remark}

\vspace{-2mm}
\begin{proof}[ Proof Sketch of Theorem \ref{thm:characterization}.]  The  proof of Theorem \ref{thm:characterization} is somewhat involved. We give a sketch here and defer formal arguments to Appendix \ref{app_sec:proof-characterization}.   The most challenging part is to find the optimal function $u^* \in \C$, which is a \emph{functional optimization} problem.    Standard optimization analysis only apply to programs with vector variables, while not functional variables. To overcome this challenge, our proof has two major steps. First, 
 we  argue that the concave functions of the specific format as in Equation \eqref{eq:optimal-format} would suffice to help the buyer to achieve optimality. This effectively    reduce the functional optimization problem to a variable optimization problem since any function of  Format  \eqref{eq:optimal-format} can be characterized by $d+1$ variables. Second, we then analyze the variable optimization problem we get and prove characterization of its optimal solutions. The first step is the most involved part and uses a significant amount of convex analysis. Such complication comes from the reasoning over the tri-level optimization problem  \eqref{eq:equilibrium-def}. Tri-level optimization is generally highly intractable  \cite{ben1990computational}  --- indeed, as we will show later,  computing the equilibrium $(u^*, \bvec{p}^*, \bvec{x}^*)$ is NP-hard in general.  Nevertheless, our analysis was able to bypass the difficulty by leveraging the special structure of the pricing problem and leads to a clean and useful characterization for the structure of the equilibrium.

  A crucial intermediate step  is the following characterization for a slightly simpler version of the question. That is,  fixing any bundle $\bar{\bvec{x}} \in X$, which imitative value function $u(\bvec{x})$ will maximize the utility of the \emph{buyer} with true value function $v(\bvec{x})$, subject to that the optimal buyer purchase response under $u(\bvec{x})$  is $\bar{\bvec{x}}$? Fortunately,  this question indeed admits a succinct characterization as shown below.
 
 \begin{restatable}{lemma}{primelemmacharacterizationfixy}
 	\label{lem:characterization-fix-y0}
 	For any bundle $\bar{\bvec{x}} \in X$, the optimal buyer imitative   value function 
 	$\bar{u}(\bvec{x})$, subject to that the resultant optimal buyer purchase response is  bundle $\bar{\bvec{x}}$, can without loss of generality have the following piece-wise linear  concave function format, parameterized by a real number $ \bar{p} \in \RR$: 
 	\begin{equation} \label{eq:optimalformat_lem0}
 		\bar{u}(\bvec{x})  = \bar{p} \cdot \min \{ \frac{x_1}{\bar{x}_1}, \cdots \frac{x_d}{\bar{x}_d}, 1\}  
 	\end{equation}
 	where $\bar{p}  \in \RR$  is the solution to the following linear program (LP):
 	\begin{lp*} \label{lp:optimalformat_lemm}
 		\maxi{ v(\bar{\bvec{x}}) - p   }
 		\st 
 		\qcon{p - c(\bar{\bvec{x}}) \geq \alpha  \cdot p - c(\alpha \bar{\bvec{x}})}{   \alpha \in [0,1] }
 	\end{lp*}
 \end{restatable}
\vspace{-6mm}
\end{proof}

%% file: instantiate.tex
\section{Explicit Characterizations for Convex and Concave Costs}\label{sec:instantiation}

In this section, we instantiate Theorem \ref{thm:characterization} to both convex and concave cost functions, arguably the most widely adopted two classes of cost functions. In both cases, we give more explicit characterizations of the equilibrium outcome, including the buyer's optimal imitative value function as well as both agents' payoffs.  




Convex production costs are widely adopted in many applications   \cite{beato1985marginal,turvey1969marginal}.  When $c(\bvec{x})$ is convex and differentiable, we show the following explicit characterization about the equilibrium outcome. A graphical visualization for this theorem is depicted in the left panel of Figure \ref{fig:one-dimensional-characterization2}.

\begin{restatable}{restatethm}{primethmcharacterizationone}
	\label{thm:characterization1} When $c$ is convex and differentiable,  the piece-wise linear concave  value function $u^*(\bvec{x})$ defined by Equation \eqref{eq:optimal-format}, with $p^* = \bvec{x}^* \cdot \nabla  c(\bvec{x}^*)$ and $\bvec{x}^* = \arg \max_{\bvec{x} \in X} ~[v(\bvec{x}) - \bvec{x} \cdot \nabla  c(\bvec{x})]$, is an  optimal buyer imitative value function. 

Under $u^*(\bvec{x})$, the trade  happens at bundle $\bvec{x}^*$ with payment $p^* = \bvec{x}^* \cdot \nabla  c(\bvec{x}^*)$.  The seller revenue $[\bvec{x}^* \cdot \nabla  c(\bvec{x}^*) - c(\bvec{x}^*)]$ is precisely the \emph{Bregman divergence} $D_{c}(\bvec{0},\bvec{x}^*)$ between $\bvec{0}$ and $\bvec{x}^*$. The buyer surplus is $[v(\bvec{x}^*) - \bvec{x}^* \cdot \nabla  c(\bvec{x}^*) ]$.
\end{restatable}

\begin{figure}[!ht]
	\centering
	\begin{subfigure}{.49\textwidth}
		\centering 
		\includegraphics[page=1,width=.5\linewidth]{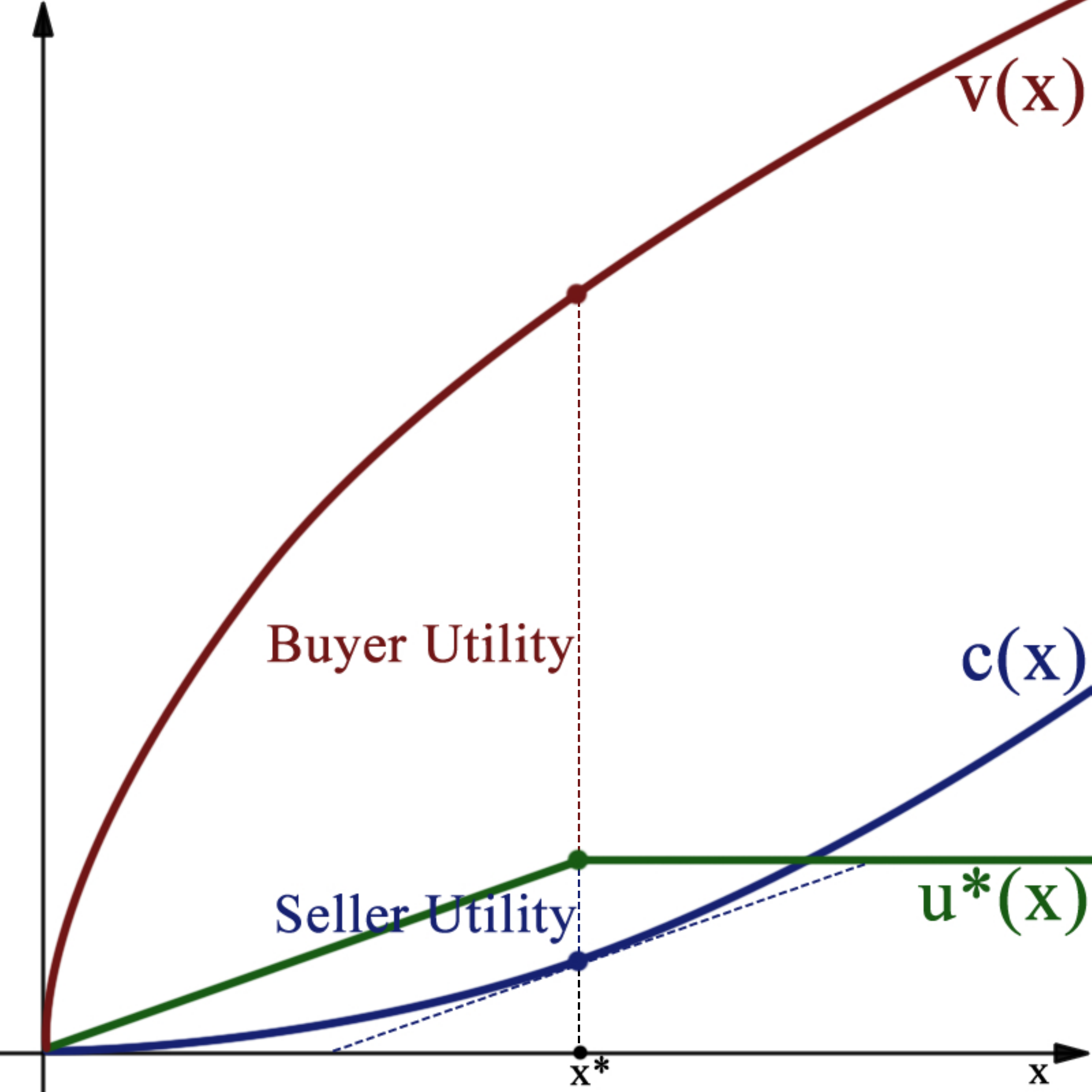} 
		\caption{An example for convex cost: $v(x) = 64 \sqrt{x}$, $c(x) = x^2$. Theorem \ref{thm:characterization1} implies $x^* = 4$, $p^* = \nabla c(x^*) = 8$, seller utility $D_c(0,4) = 16$ \\and buyer utility $v(x^*) -  p^* \cdot  x^* = 96$.}
		\label{fig:sub1}
	\end{subfigure}
	\begin{subfigure}{.49\textwidth}
		\centering
		\includegraphics[page=1,width=.5\linewidth]{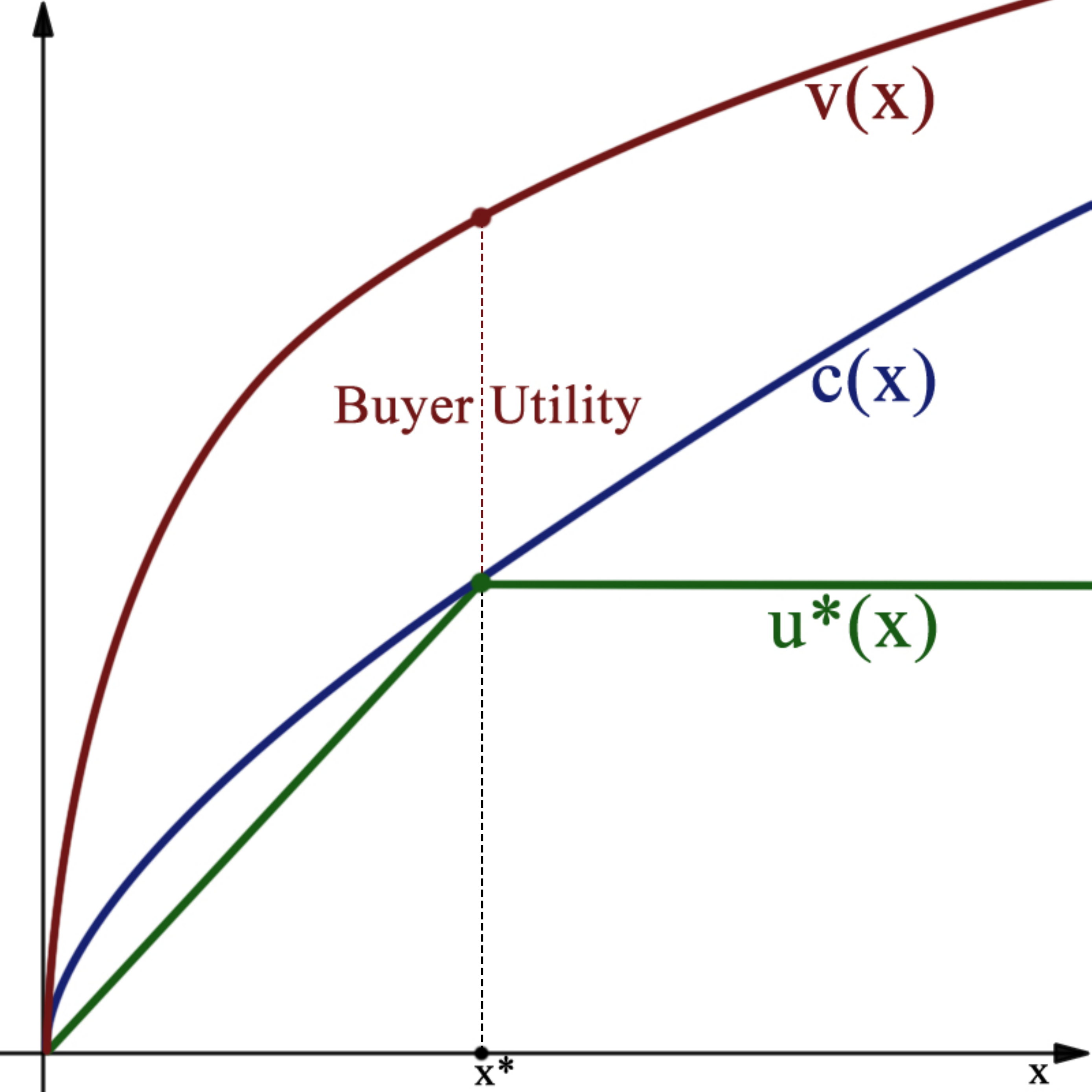} 
		\caption{An example of one-dimensional concave costs: $v(x) = 4  x^{\frac{1}{4}}$, $c(x) = \sqrt{x}$. Theorem \ref{thm:characterization2} implies $x^* = 16$, $p^* = \frac{c(x^*)}{x^*}= 1/4$, seller utility is $0$ and buyer utility $v(x^*) -  p^* \cdot x^* = 4$.}
		\label{fig:sub2}
	\end{subfigure}
	\caption{Illustration of the equilibrium characterization:  $v(x)$ is buyer's true value function, $c(x)$ is  the seller's cost function, and $u^*(x)$ is the  optimal buyer imitative value function. 
\vspace{-2mm}	
} 
	\label{fig:one-dimensional-characterization2}
\end{figure}

We now consider  concave costs and prove the following characterization. A graphical visualization for the theorem is depicted in the right panel of Figure \ref{fig:one-dimensional-characterization2}. 

\begin{restatable}{restatethm}{primethmcharacterizationtwo}
	\label{thm:characterization2} When  $c(\bvec{x})$ is     concave, the piece-wise linear concave  value function $u^*(\bvec{x})$ defined by Equation \eqref{eq:optimal-format}, with $p^* =  c(\bvec{x}^*)$  and $\bvec{x}^* =\arg \max_{\bvec{x} \in X} ~[v(\bvec{x}) - c(\bvec{x})]$, is an  optimal buyer imitative value function.  

Under $u^*(\bvec{x})$, the trade  happens at bundle $\bvec{x}^*$ with payment $p^* =  c(\bvec{x}^*)$.  The seller revenue will be $0$. The buyer extracts the maximum possible surplus   $ \max_{\bvec{x} \in X} [v(\bvec{x}) -    c(\bvec{x}) ]$. 
\end{restatable}

\vspace{-2mm}
Graphical visualizations for the above two theorems are depicted in Figure \ref{fig:one-dimensional-characterization2}.  Note that both examples in Figure \ref{fig:one-dimensional-characterization2} show the sub-optimality of imitating $u^*(x) = c(x)$ for the buyer. To see this, note  $u^*(x) = c(x)$ makes the game zero-sum. To guarantee non-negative revenue, the seller must pick a price $p$ such that the line $p\cdot x$ is always above $u^*(x)$. A buyer imitating $u^*(x) = c(x)$ will end up purchasing a $0$ amount in both cases, and thus are sub-optimal for him.

An important conceptual message from Theorem \ref{thm:characterization2} is that when the seller production cost is concave and known to the buyer, the buyer can always come up with  an imitative value function which squeeze the seller's revenue to its extreme, i.e., $0$.\footnote{In real practice, the buyer can slightly deviate from his value function to given a small  $\epsilon$ amount of incentive for the seller to strictly prefer production.}  Under the seller's optimal imitative value function, the trade will happen at his most favorable bundle amount $\bvec{x}^*$ and  the seller pays just the cost function $c(\bvec{x}^*)$ to the seller.  

Theorem \ref{thm:characterization2}  is certainly bad news for the seller.  However, it is a \emph{descriptive}  result and only shows it is possible for the buyer to achieve maximum possible surplus.  Our next result brings somewhat good news to the seller. Specifically, we show it is NP-hard for the buyer to optimize, even approximately,  his optimal surplus. The hardness holds even when the production cost function $c(\bvec{x})$ is concave, in which case the surplus is characterized by Theorem \ref{thm:characterization2}. This shows that even though in theory the buyer can derive large surplus from strategic manipulation, even approximately figuring out such an optimal manipulation is  impossible in general, unless P = NP.

\begin{restatable}{restatethm}{primethmhardness}
	\label{prop:hardness}[Intractability of Equilibrium] 
	It is NP-hard  to approximate the buyer equilibrium  surplus in \name{}  games to be within ratio $1/d^{1- \epsilon}$ for any $\epsilon >0$. This hardness result holds  even when    the production cost function $c(\bvec{x})$ is concave and  the buyer's true value function $v(\bvec{x}) $ is simply the linear function $ \sum_{i=1}^d x_i$.  
\end{restatable}

Formal proofs for both equilibrium characterizations and the hardness of approximating the buyer equilibrium surplus can be found in Appendix \ref{app_sec:proof-section-initantiation}. We note that an intriguing open question is whether optimal imitative value function can be computed for \emph{convex} production cost $c(\bvec{x})$.  In this case, the optimal bundle $\bvec{x}^* $ can be explicitly expressed as $\bvec{x}^* = \argmax_{\bvec{x}} [ u(\bvec{x}) - \bvec{x} \cdot \nabla c(\bvec{x})]$  for any value function $u$, however how to derive the optimal $u^*$ remains challenging.   

%% file: concave.tex
\section{The Risk of Over-exploiting Buyer's Incentives}\label{sec:nonlinear-pricing}
To counteract the buyer's strategic manipulation, one of the most  natural approaches  is perhaps to use a  richer or more powerful  class of pricing schemes, as opposed to  only using linear pricing.  Such additional power of pricing will always increase the revenue when facing an  honest buyer. Unfortunately, however, we show that it   does not necessarily help in the presence of buyer manipulation --- in fact, the seller may suffer the risk of \emph{over-exploiting the buyer's incentives} so that it becomes easier for the buyer to manipulate.  This phenomenon is similar in spirit to the \emph{overfitting} phenomenon in machine learning. That is,   using a richer hypothesis class does not necessarily reduce the testing error, though it does reduce the training error.  We believe that our findings in this section   partially explain why simple pricing schemes like linear pricing are  preferred in reality.   

We first prove that the strictly more general class of  \emph{concave} pricing schemes can never do better for the seller than the (much) restricted class of linear pricing. In fact, surprisingly,  the equilibrium of \name{} under concave pricing turns out to be exactly the same as the equilibrium under linear pricing. 
This time, our proof utilizes  a crucial observation that under concave pricing, the tri-level optimization problem of \name{}  can be reduced to solving a bi-level optimization problem (in particular, FOP \eqref{lp:optimal_deception_concave_pricing}). Recall that the proof of Theorem \ref{thm:characterization}  also reduces the tri-level FOP to a bi-level FOP through a characterization about the price and optimal buyer bundle in Lemma \ref{lem:optimalpricing}.  However,    Lemma \ref{lem:optimalpricing} does not hold any more  if   the seller uses the richer class of concave pricing schemes. Therefore,  the FOP \eqref{lp:optimal_deception_concave_pricing} we obtain here is different from the core FOP \eqref{lp:optimaldeception} we solve in the proof of Theorem \ref{thm:characterization}. Nevertheless, through careful convex analysis, we are able to show  that the optimal solution to FOP \eqref{lp:optimal_deception_concave_pricing} also admits an optimal solution of similar structure  as  characterized by Theorem \ref{thm:characterization}.  



\begin{restatable}{restatethm}{primethmnonlinearpricing}
 	\label{thm:non_linear_pricing_eqm}
The equilibrium of  \name{}  under \emph{concave} pricing is exactly the same as the equilibrium under \emph{linear} pricing as characterized by Theorem \ref{thm:characterization}. 
\end{restatable}
\begin{proof}[Proof Sketch of Theorem \ref{thm:non_linear_pricing_eqm}.] 
We start by examining  how the use of  concave pricing schemes may simplify FOP \eqref{eq:equilibrium-def}. Recall that both value functions and pricing functions are monotone non-decreasing and normalized to be $0$ at  $\bvec{0}$. For any concave buyer value function $u \in \C$, it is easy to see that the optimal pricing function $p$ simply equals $u$ (i.e., charging buyer exactly his imitative value) and ask the buyer to break ties in favor of the seller by picking $\bvec{x}^* = \arg\max_{\bvec{x} \in X} [u^*(\bvec{x}) - c(\bvec{x})]$ to maximize the seller's revenue. Consequently, the use of concave pricing simplifies FOP \ref{eq:equilibrium-def} to the following bi-level functional optimization problem: 
\begin{eqnarray}\nonumber 
  u^* = \argmax_{u  \in \C}  [  v(\bvec{x}^*) - p^*(\bvec{x}^* ) ],    
\text{ where }   \bvec{x}^* = \argmax_{\bvec{x} \in X} \big[ u^*(\bvec{x}) - c( \bvec{x} ) \big] 
\end{eqnarray} 

We fix a particular bundle $\bar{\bvec{x}}$ and  examine   what buyer value function $u \in \C$ would maximize the buyer's utility subject to that the trade will happen at bundle $\bar{\bvec{x}}$.  This results in the following \emph{functional optimization problem} (FOP) for the buyer with \emph{functional variable} $u$. 
\begin{lp}\label{lp:optimal_deception_concave_pricing}
	\maxi{ v(\bar{\bvec{x}}) - u(\bar{\bvec{x}}) }
	\st 
	\qcon{ u(\bar{\bvec{x}}) - c(\bar{\bvec{x}}) \geq u(\bvec{x}') - c(\bvec{x}') }{ \bvec{x}' \in X }
\end{lp}
where the constraint means the seller's optimal price  for  value function $u$ is $u(\bvec{\bar{x}})$ and thus the buyer best response amount is indeed $\bar{\bvec{x}}$.   The remaining proof relies primarily on the following lemma. 
	\begin{restatable}{lemma}{primelemmanonlinearpricing}
	\label{lemma:non_linear_pricing}
 The following concave function is optimal to FOP \eqref{lp:optimal_deception_concave_pricing}: 
  \begin{equation}\label{eq:opt-format-concave}
 \bar{u}(\bvec{x}) =\bar{p}  \cdot \min \{ \frac{x_1}{\bar{x}_1}, \cdots \frac{x_d}{\bar{x}_d}, 1\},  \quad \text{ where }  \bar{p} = \sup_{\alpha \in [0,1)} \frac{  c(\bar{\bvec{x}}) -c(\alpha \bar{\bvec{x}}) }{1 - \alpha}. 
 \end{equation}
\end{restatable}
Given this characterization, the buyer simply needs to look for the best $\bar{\bvec{x}}$. This then leads to the same characterization as in Theorem \eqref{thm:characterization} since Equation \eqref{eq:opt-format-concave} is the same as the value function characterized in Equation \eqref{eq:optimalformat_lem0}. The proof of  Lemma \ref{lemma:non_linear_pricing} is deferred to Appendix \ref{app_sec:proof-non-linear-pricing}.
\end{proof}
	
Theorem \ref{thm:non_linear_pricing_eqm} shows that more general class of pricing schemes may not help the seller to obtain more revenue. One might then wonder whether it at least never hurts since if that is the case, at least it would  never be a worse choice.  Unfortunately, 
our following example   shows that  a richer class of pricing schemes may   bring \emph{strict harm} to the seller and \emph{strict benefit} to the seller.  

\begin{example}[ The Risk of Overexploiting Buyer Incentives] \label{ex:hurt}
 There is a single type of divisible good to sell, i.e., $d=1$. Let the seller's production cost function be the convex function $c(x) = x^2$ and let the buyer's true value function $v(x)$ be the following piece-wise linear concave function  
\[
v(x) =
\begin{cases}
10x & 0 \leq x \leq 0.81 \\
8.1 & x > 0.81 
\end{cases}
\] 
\end{example}

Let $\P_C$ denote  the set of all  concave pricing schemes. The following proposition completes Example \ref{ex:hurt} and its proof can be found in Appendix \ref{app_sec:proof_concave_sub_class}. 
\begin{restatable}{proposition}{primepropsubclass}
	\label{thm:concave_sub_class}
For the instance   in Example \ref{ex:hurt},  there exists  pricing scheme class $\P $ with $\P_L \subset \P \subset \P_C$  such that when the seller changes from  linear pricing  class  $\P_L$ to the richer  class   $\P$, the seller's revenue strictly \emph{decreases} and the buyer's surplus strictly \emph{increases} at the equilibrium of \name{}.  
\end{restatable}

%% file: conclude.tex
\section{Conclusions}\label{sec:conclusion}

Motivated by optimal pricing against an unknown buyer, this paper put forwards a simple variant of the very basic pricing model by augmenting it with an additional stage of buyer commitment at the beginning. This motivation is driven by the seller's ignorance of the buyer's value function and thus have to price in the dark. We fully characterize the equilibrium of this new game model. The equilibrium reveals interesting insights about what the seller can learn, and how much seller revenue and   buyer surplus it may result in. We also show that more general class of pricing schemes may overfit the buyer's incentive and lead  to a pricing game that is even easier for the buyer to manipulate. 


Our results opens the possibilities for many other interesting questions. For example, given the risk of using a too general class of pricing schemes, what class of pricing schemes is a good compromise between extracting revenue and  robust to buyer manipulations? Is linear pricing scheme the best such class or any other pricing scheme?  Second, as the first study of our setup, we have chosen to  focus  on a simple setup where the seller has completely no knowledge about the buyer's value function. An interesting question is, how the seller's learning and resultant revenue may increase when the seller has some prior knowledge about the buyer's values. In fact, one natural modeling question  is how to model the seller's prior knowledge about the buyer's  value function. Is the prior knowledge   about which subclass the value functions are from or  about what distribution the parameters of the value functions are from?  Finally, our model assumes that the buyer has full knowledge about the seller. An ambitious though extremely intriguing question to ask is what if the buyer also has incomplete knowledge about the seller and how to analyze the equilibrium under the information asymmetry from both sides.

%% file: appendix.tex

\section{Additional Discussions} 
\subsection{Additional Discussions on Related Works}\label{append:sec:related}
Our work is related to a recent study by  Tang and Zeng \cite{tang2018price} who study   the bidders' problem of committing to a fake type distribution in auctions and acting consistently as if the  bidder's  type were from the fake type distribution. This is similar in spirit to our buyer's commitment to an imitative value function. However, there are two key differences between our work and \cite{tang2018price}.  First, the seller in our model (realistically) has production cost where as the auction setting of \cite{tang2018price}  does not have production cost. This is an important difference because with $0$ production cost, the optimal manipulation in our case is trivial,  i.e. the buyer will imitate a value function of 0.  However, this trivial solution does not arise in the model of \cite{tang2018price} because in their setup there are multiple buyers (bidders) and the competition among bidders increases the auctioneer's revenue despite bidders' imitative or faking behaviors. This is the second key difference between our work and \cite{tang2018price}. Therefore,  our work illustrates  how the production cost can affect the buyer's strategic manipulation and the seller's revenue whereas the work of Tang and Zeng \cite{tang2018price}  sheds light on how the bidders' competition affect the auctioneer's mechanism design and ultimate revenue. Though these two aspects are not comparable,  we  believe they are both interesting for a deep understanding.

Another  very relevant literature is  learning the optimal prices or optimizing aggregated total revenue by repeatedly  interacting with a single buyer   \cite{amin2013learning,amin2014repeated,mohri2014optimal,mohri2015revenue,vanunts2019optimal}.  Similar to us,  they also consider the buyer's strategic behavior  that potentially tricks the seller's learning algorithm. However,   these previous works all focus on designing learning algorithms that can handle strategic data sources. Our work can be viewed as a complement to this literature. Instead of proposing new algorithms, we focus on understanding the limits of what learning algorithms can achieve by analyzing a basic model which is a   variant of the textbook-style optimal pricing model.  Moreover, the setups of these previous works are also different from us, which make them not comparable to us. For example,  some of these models \cite{amin2013learning,amin2014repeated,mohri2015revenue}  assume that the buyer's values for goods  are drawn from  distributions (a.k.a., demand distribution) and consequently his best responses to seller prices are  stochastic with randomness inherited from his value distribution. 
However,  our model  assumes that  the buyer has an unknown but fixed value function that drives his  purchase responses.  Such a response is the solution to buyer's optimization problem while not from a random distribution.   There have also been models that  consider unknown but fixed buyer values like us \cite{mohri2014optimal,vanunts2019optimal}. However, these works  have focused on a single indivisible good  with discrete agent utilities whereas  our model has \emph{multiple divisible} goods with continuous agent utilities. 

There have also been studies on learning the optimal prices from \emph{truthful} revealed preferences, i.e., assuming the buyer will honestly best respond to seller prices \cite{beigman2006learning,zadimoghaddam2012efficiently,balcan2014learning,pmlr-v119-zhang20f,amin2015online,roth2016watch,roth2020multidimensional}. Our works try to understand if the assumption of truthful revealed preferences does not hold and if the buyer will strategically respond to the seller's learning, what learning outcome could be expected when the buyer simply imitates a different value function that is  optimally chosen. From this perspective, these works serve as a key motivation for the present paper.

More generally, our work subscribes to the general line of research on learning from strategic data sources.  Strategic classification has been studied in other different settings or domains or for different purposes, including spam filtering  \cite{bruckner2011stackelberg},    classification under incentive-compatibility constraints \cite{zhangincentive},  online learning~\cite{dong:ec18,chen2020learning}, and understanding the social implications~\cite{akyol2016price,milli:fat19,hu:fat19}.  Finally, going beyond classification, strategic behaviors in machine learning has received significant recent attentions, including in regression problems \cite{Perote2004StrategyproofEF,dekel2010incentive,chen:ec18}, distinguishing distributions \cite{zhang2019distinguishing,zhang2019samples}.

\subsection{Additional Discussion  on Buyer's Commitment}\label{app_sec:commitment}
The buyer's ability of making such a commitment  fundamentally comes from the fact that the seller has no prior knowledge about the buyer's true value function $v(\bvec{x})$, i.e., has to ``price in the dark''.\footnote{By convention, all value functions are assumed to be concave.  This  can  be  equivalently   viewed as a restriction to the buyer's manipulation imposed by the seller's (very limited) prior  knowledge about concavity of buyer values. Later, we will  also briefly discuss how the absence of this prior knowledge may lead to   worst seller revenue. }  To find a good pricing scheme, the seller may   interact with the buyer to learn the his value function \cite{zadimoghaddam2012efficiently, beigman2006learning,balcan2014learning}, or  learn the optimal pricing scheme \cite{roth2016watch}, or directly optimize the aggregated revenue during repeated interactions \cite{amin2013learning,amin2014repeated,mohri2014optimal,mohri2015revenue,vanunts2019optimal}.  However, regardless what algorithm the seller may adopt, the buyer can always choose to consistently behave according to a carefully crafted different  value function $u(\bvec{x})$, i.e., the \emph{commitment}.  For example, suppose the seller tries to apply any machine learning algorithm, the buyer may respond by  directly announcing  his imitative value function $u(\bvec{x})$ even before the learning starts and then behave consistently. In such scenarios,   learning is even not needed since the best a seller can do is to respond with the optimal pricing against $u(\bvec{x})$. Similarly, if the seller adopts any  dynamic pricing mechanism, the buyer may respond similarly by announcing her value function $u(\bvec{x})$.   Since the seller  lacks  knowledge about the buyer's value, such imitative buyer behavior  makes him indistinguishable from a buyer who truly has value function $u(\bvec{x})$.  Therefore, our equilibrium characterization in later sections will help to understand   what   the optimal imitative value function for the buyer is and what revenue   the seller can possibly achieve when pricing against such a  strategic buyer in the dark, i.e., without any prior knowledge.

We remark that though the imitative strategy may not always be the absolutely best possible strategy for the buyer, but it  enjoys many advantages. First of all,  as we will prove later, this strategy does lead to significant improvement to the buyer's utility\footnote{It is never worse since the buyer can always at least behave truthfully by letting $u(\bvec{x}) = v(\bvec{x})$. }  and, in fact, is provably  the best possible (among all possible strategic behaviors that the buyer may  adopt)  in certain circumstances, e.g., when the seller's production cost function $c(\bvec{x})$ is concave. 
Second, this imitative strategy is easy to adopt in practice and requires no knowledge about the seller pricing scheme. For example,  this imitative strategy  works equally well for any price learning algorithm so long as it can effectively learn the optimal price from the buyer.  Third, it also has good long term effect since  the seller cannot distinguish whether the buyer truly has value function $u(\bvec{x})$ or not, and may just have to use the same learned price for  future purchases from this buyer. However, we  show later that any buyer behavior that is not consistently imitating a value function can be easily identified by the seller.  In such situations, even though the seller ended up with some prices, she knows that it is not the  truly optimal price for this buyer and may take this into account in future interactions.

\section{Technical Background: Concave/Convex Functions and Super/Sub-Gradients}\label{app_sec:convex_background}
Let  $f(\bvec{x}): X \to \RR$ be any function where $X \subset \RR^d$ is the domain of $f$.  A vector $\bvec{p} \in \RR^d$ is called a \emph{super-gradient} for $f$ at $\bvec{x} \in X$ if for any $\bvec{z} \in X$ we have $f(\bvec{z}) \leq f(\bvec{x}) + \bvec{p} \cdot (\bvec{z} - \bvec{x})$.   Function $f$ is called \emph{concave} if  for any $\bvec{x}, \bvec{z} \in \RR^d$ and any $\alpha \in [0,1]$ we have $ \alpha f(\bvec{x}) + (1-\alpha) f(\bvec{z}) \leq f(\alpha \bvec{x} +(1-\alpha) \bvec{z})$.  Super-gradients do not always exist. However, a concave function has at least one super-gradient at any $\bvec{x} \in X$. For a differentiable concave function $f$,  its  gradient $\nabla f(\bvec{x})$ is the only super-gradient at $\bvec{x}$ for any $\bvec{x}\in X$.  If $f$ is concave but not differentiable, it may have multiple super-gradients at some $\bvec{x}$. In this case, we use $\partial   f(\bvec{x})$ to denote the \emph{set} of all super-gradients of $f$ at $\bvec{x}$. 
Among all super-gradients in $\partial   f(\bvec{x})$, of our particular interest is the following one:  $\nabla_{\max} f(\bvec{x}) = \argmax_{\bvec{p} \in \partial   f(\bvec{x})} [\bvec{x} \cdot \bvec{p}]$.  This is  the super-gradient that maximize linear function $\bvec{p} \cdot \bvec{x}$.  
When $f$ is differentiable, $\nabla_{\max} f(\bvec{x}) =\nabla f(\bvec{x}) $  is the (only) super-gradient.

Function $f$ is called \emph{convex} if $-f$ is concave. 
For convenience of stating our results, we will mostly work with \textit{differentiable} convex functions in this paper.    A useful distance notion for differentiable convex function $f$  is the \emph{Bregman divergence}: $$D_{f}(\bvec{z},\bvec{x}) = f(\bvec{z}) - f(\bvec{x}) - \nabla f(\bvec{x})\cdot [  \bvec{z}- \bvec{x}].$$
$D_{f}(\bvec{z},\bvec{x}) $ is always non-negative for convex functions and strictly positive for strictly convex functions when  $\bvec{z} \not = \bvec{x}$. However, Bregman divergence is asymmetric among variables, i.e., $D_{f}(\bvec{z},\bvec{x})  \not = D_{f}(\bvec{x},\bvec{z}) $ in general.

\section{Proof of Theorem \ref{thm:characterization}}\label{app_sec:proof-characterization}

\primethmcharacterization*

  We start with a useful lemma that characterizes the relation between optimal seller price $\bvec{p}_u$ and optimal buyer bundle $\bvec{x}_u$ for any concave buyer utility function $u(\bvec{x})$.  A similar result  has been proved in \cite{roth2016watch}. The only difference here is that we allow any concave buyer utility function whereas the buyer utility function in \cite{roth2016watch} is  assumed to be strictly concave and differentiable.  Nevertheless, the proof remains similar and  thus is omitted due to space limit. 
 
 \begin{lemma} \label{lem:optimalpricing}
 	For any concave buyer value function $u(\bvec{x}): \mathbb{R}^{d}_+ \rightarrow \mathbb{R}_+$   reported by the buyer, let $\bvec{p}_u$ be the optimal price vector for the seller and  $\bvec{x}_u$ be the resultant buyer optimal bundle for purchase, then the following relation holds:
 	\begin{equation} \label{eq:optimalpricing} 
 		\bvec{p}_u = \nabla_{\max} u(\bvec{x}_u).    
 	\end{equation} 
 	where $\nabla_{\max} u(\bvec{x}) = \argmax_{\bvec{p} \in \partial   u(\bvec{x})} [\bvec{x} \cdot \bvec{p}]$.\footnote{The ``$\max$" comes from the fact that the seller will pick the profit-maximizing price if multiple prices  result in the same optimal buyer purchase.}
 \end{lemma}

 A crucial intermediate step in our proof is the following characterization for a slightly simpler version of the question. That is,  fixing any bundle $\bar{\bvec{x}} \in X$, which imitative value function $u(\bvec{x})$ will maximize the utility of the \emph{buyer} with true value function $v(\bvec{x})$, subject to that the optimal buyer purchase response under $u(\bvec{x})$  is $\bar{\bvec{x}}$?   If we can find a succinct characterization for this simpler question, what remains is just to search for the best bundle $\bar{\bvec{x}}$.  That will be a  (variable) optimization problem,  which is suitable for standard optimization techniques to solve.  Fortunately,  the above question does admit a succinct characterization. 
 
 \begin{lemma}\label{lem:characterization-fix-y}[Lemma \ref{lem:characterization-fix-y0} restated]
 	\, For any bundle $\bar{\bvec{x}} \in X$, the optimal buyer imitative   value function 
 	$\bar{u}(\bvec{x})$, subject to that the resultant optimal buyer purchase response is  bundle $\bar{\bvec{x}}$, can without loss of generality have the following piece-wise linear  concave function format, parameterized by a real number $ \bar{p} \in \RR$: 
 	\begin{equation} \label{eq:optimalformat_lem}
 		\bar{u}(\bvec{x})  = \bar{p} \cdot \min \{ \frac{x_1}{\bar{x}_1}, \cdots \frac{x_d}{\bar{x}_d}, 1\}  
 	\end{equation}
 	where $\bar{p}$  is the solution to the following linear program (LP):
 	\begin{lp}\label{lp:buyer-simple-lemma}
 		\maxi{ v(\bar{\bvec{x}}) - p   }
 		\st 
 		\qcon{p - c(\bar{\bvec{x}}) \geq \alpha  \cdot p - c(\alpha \bar{\bvec{x}})}{   \alpha \in [0,1] }
 	\end{lp}
 	Moreover, under imitative value function $\bar{u}(\bvec{x})$, we have
 	\begin{enumerate}
 		\item For any convex coefficients $\lambda \in \Delta_d$,  the linear pricing scheme with unit price vector   $ (\lambda_1 \cdot \frac{\bar{p}}{\bar{x}_1}, \lambda_2 \cdot \frac{\bar{p}}{\bar{x}_2}, \cdots, \lambda_d \cdot \frac{\bar{p}}{\bar{x}_d})$ will be optimal. 
 		\item In any of the above optimal linear pricing schemes, the buyer's optimal  bundle response will always be  $\bar{\bvec{x}}$ and the buyer payment will equal $\bar{p}$.    
 	\end{enumerate}
 \end{lemma}
 \begin{proof}[Proof of Lemma \ref{lem:characterization-fix-y}]
 	From the buyer's perspective,  with a fixed bundle $\bar{\bvec{x}}$  in mind,  his problem is to come up with an imitative value function $u(\bvec{x})$ such that its corresponding price $\nabla_{\max} u(\bar{\bvec{x}})$ as from Lemma \ref{lem:optimalpricing} maximizes his revenue at bundle $\bar{\bvec{x}}$.   This results in the following \emph{functional optimization problem} (FOP) for the buyer with \emph{functional variable} $u$. 
 	\begin{lp}\label{lp:optimaldeception}
 		\maxi{ v(\bar{\bvec{x}}) - \bar{\bvec{x}} \cdot \nabla_{max} u(\bar{\bvec{x}}) }
 		\st 
 		\qcon{ \nabla_{max} u(\bar{\bvec{x}}) \cdot \bar{\bvec{x}} - c(\bar{\bvec{x}}) \geq \nabla_{max} u(\bvec{x}') \cdot \bvec{x}' - c(\bvec{x}') }{ \bvec{x}' \in X }
 		\con{u \text{ is concave}}
 	\end{lp}
 	where the first constraint means the seller's optimal price for  value function 
 	$u(\bvec{x})$ is $\nabla_{max} u(\bar{\bvec{x}}) $ and thus the buyer best response bundle is indeed $\bar{\bvec{x}}$.  
 	
 	The lemma states that the $\bar{u}(\bvec{x})$ defined in Equation \eqref{eq:optimalformat_lem} is an  optimal solution to FOP \eqref{lp:optimaldeception}. To analyze FOP \eqref{lp:optimaldeception}, we first simplify the class of concave function $u$ that we need to consider. In particular, we claim that there always exists an optimal solution to FOP \eqref{lp:optimaldeception} such that $u(\bvec{x})$ has the following   form:
 	\begin{equation} \label{optimalformat2}
 		u(\bvec{x}) =p  \cdot \min \{ \frac{x_1}{\bar{x}_1}, \cdots \frac{x_d}{\bar{x}_d}, 1\}  
 	\end{equation}
 	where  the only \emph{parameter} $p \in \RR$  is the  coefficient vector. To prove this, let $u^*(\bvec{x})$ be any optimal solution to FOP  \eqref{lp:optimaldeception}.  Construct another concave value function $\bar{u}(\bvec{x})$ as follows:
 	
 	\begin{equation} \label{optimalformat_specific}
 		\bar{u}(\bvec{x}) =\bar{p} \cdot \min \{ \frac{x_1}{\bar{x}_1}, \cdots \frac{x_d}{\bar{x}_d}, 1\} , \text{ where } \bar{p} = \nabla_{\max} u^* (\bar{\bvec{x}}) \cdot \bvec{\bar{x}}  
 	\end{equation}
 	
 	%
 	Note that $\bar{p}$ is precisely the payment for  a buyer with value function $u^*$ when his optimal bundle amount is   $\bar{\bvec{x}}$.  We show that this new value function will result in the same  optimal buyer bundle $\bar{\bvec{x}}$ and payment  $\bar{p}$. It thus does not change either of the agent's utilities and remains optimal for the buyer.  

 	First, we argue that the constructed $\bar{u}$ is still feasible to FOP \eqref{lp:optimaldeception}.  Concavity of $\bar{u}$ is evident since it is the minimum of a set of linear functions. 
 	By Lemma \ref{lem:optimalpricing},  we have $\nabla_{\max} \bar{u}(\bvec{x}) = 0$  if $\bvec{x}$ is element-wise strictly greater than $\bvec{\bar{x}}$. Otherwise, let $i^* = \argmin_{\{i|x_i \leq \bar{x}_i\}} \frac{x_i}{\bar{x}_i} $,\footnote{If there are multiple $i$ that all minimize $\frac{x_i}{\bar{x}_i} $, the proof is valid by picking any of them. }  we have $[\nabla_{\max} \bar{u}(\bvec{x})]_{i^*} = \frac{\bar{p}}{\bar{x}_{i^*}} $ which is the $i^*$'th element of $\nabla_{\max} \bar{u}(\bvec{x})$, while all the other gradient entries of $\nabla_{\max} \bar{u}(\bvec{x})$ are $0$. Therefore, we have $ \nabla_{\max} \bar{u}(\bvec{x}) \cdot \bvec{x} = \bar{p}   \cdot  \frac{x_{i^*}}{\bar{x}_{i^*}} $. Specifically, $\nabla_{\max} \bar{u}(\bar{\bvec{x}}) \cdot \bvec{\bar{x}}=\bar{p}$ which equals precisely the buyer payment $ \nabla_{\max} u^*(\bar{\bvec{x}}) \cdot \bvec{\bar{x}}$ under utility $u^*$.  
 	
 	We now verify that the constraints in FOP \eqref{lp:optimaldeception} still holds for $\bar{u}$.  We start from verifying for special $\bvec{x}'$s where there exists $\alpha \in [0,1]$   such that $\bvec{x}' =\alpha \cdot \bar{\bvec{x}}$.  In this case, we have 
 	\begin{align*} 
 		\nabla_{\max} \bar{u}(\bar{\bvec{x}}) \cdot \bar{\bvec{x}} - c(\bar{\bvec{x}}) &= \nabla_{\max} u^*(\bar{\bvec{x}}) \cdot \bar{\bvec{x}} - c(\bar{\bvec{x}}) & \mbox{(since $\nabla_{\max} \bar{u}(\bar{\bvec{x}}) \cdot \bvec{\bar{x}}= \bar{p} =  \nabla_{\max} u^*(\bar{\bvec{x}}) \cdot \bvec{\bar{x}}$)}\\  
 		& \geq \nabla_{\max} u^*( \bvec{x}') \cdot \bvec{x}'  - c(\bvec{x}') & \mbox{(by feasibility of $u^*$ ) }\\ 
 		& \geq \nabla_{\max} u^*(\bar{\bvec{x}}) \cdot \bvec{x}'  - c(\bvec{x}') & \mbox{(by concavity of $u^*$) }\\  
 		& = \alpha \cdot \bar{p}  - c(\bvec{x}') & \mbox{(since $\bvec{x}' = \alpha \cdot \bvec{\bar{x}}$ and $\bar{p} =\nabla_{\max} u^*(\bar{\bvec{x}}) \cdot \bvec{\bar{x}}$) } \\  
 		& = \nabla_{\max} \bar{u}(\bvec{x}') \cdot \bvec{x}' - c(\bvec{x}'). & \mbox{(by definition of $\bar{u}$ and $\bvec{x}' = \alpha \cdot \bar{\bvec{x}}$) } 
 	\end{align*}	
 	Specifically, the second inequality   holds because $\nabla_{\max} u^*(\alpha\bar{\bvec{x}}) \cdot \bar{\bvec{x}}/|\bar{\bvec{x}}|$ is the directional derivative of $u^*$ at $\alpha\bar{\bvec{x}}$ in the direction of $\bar{\bvec{x}}$, which is non-increasing with respect to $\alpha$ due to the concavity of $u^*$.
 	%
 	The above argument also implies $ \nabla_{\max} \bar{u}(\bar{\bvec{x}}) \cdot \bar{\bvec{x}} - c(\bar{\bvec{x}}) \geq 0$ by instantiating $\bvec{x}' = \bvec{0}$.  
 	
 	Next, we consider the case when $\bvec{x}' \neq \alpha \cdot \bar{\bvec{x}}$ for $\alpha \in [0,1]$.  There will be two possible situations to consider:
 	\begin{enumerate}
 		\item If $\bvec{x}' $ is element-wise greater than $\bvec{\bar{x}}$, then $ \nabla_{\max} \bar{u}(\bvec{x}') = 0$. Thus, we have $\nabla_{\max} \bar{u}(\bar{\bvec{x}}) \cdot \bar{\bvec{x}} - c(\bar{\bvec{x}}) \geq 0 > \nabla_{\max} \bar{u}(\bvec{x}') \cdot \bvec{x}' - c(\bvec{x}') = 0 - c(\bvec{x}')$.
 		\item If $\bvec{x}' $ is \textit{not} element-wise greater than $\bvec{\bar{x}}$, let $i^* = \argmin_{\{i|x'_i \leq \bar{x}_i\}} \frac{x'_i}{\bar{x}_i} $ and $\alpha^* = \min_{\{i|x'_i \leq \bar{x}_i\}} \frac{x'_i}{\bar{x}_i} $ . Then we have  $[\nabla_{\max} \bar{u}(\bvec{x}')]_{i^*} = \frac{ \bar{p}}{\bar{x}_{i^*}} $ which is the $i^*$'th element of $\nabla_{\max} \bar{u}(\bvec{x}')$, while all the other elements of $\nabla_{\max} \bar{u}(\bvec{x}')$ are $0$. In addition, denote $\hat{\bvec{x}}' = \alpha^* \cdot \bar{\bvec{x}}$. Note that for $\bvec{x}' $, we have  $\nabla_{\max} \bar{u}(\bvec{x}') \cdot \bvec{x}' = \alpha^* \cdot \bar{p} =  \nabla_{\max} \bar{u}(\hat{\bvec{x}}') \cdot \hat{\bvec{x}}'$. However, we have $c(\hat{\bvec{x}}') \leq c(\bvec{x}')$ because $\hat{\bvec{x}}'$ is element-wise less than or equal to $\bvec{x}'$. Thus, we have $\nabla_{\max} \bar{u}(\hat{\bvec{x}}') \cdot \hat{\bvec{x}}' -c(\hat{\bvec{x}}')  \geq \nabla_{\max} \bar{u}(\bvec{x}') \cdot \bvec{x}' -  c(\bvec{x}') $ by monotonicity of $c(\bvec{x})$. Our previous derivation for the special case $\bvec{x}' = \alpha \bar{\bvec{x}}$  with $\alpha \in [0,1]$ implies $	\nabla_{\max} \bar{u}(\bar{\bvec{x}}) \cdot \bar{\bvec{x}} - c(\bar{\bvec{x}}) \geq \nabla_{\max} \bar{u}(\alpha \cdot \bvec{\bar{x}}) \cdot (\alpha \cdot \bvec{\bar{x}}) - c(\alpha \cdot \bvec{\bar{x}}), \forall \alpha \in [0,1]$. These together imply  $$ \nabla_{\max} \bar{u}(\bar{\bvec{x}}) \cdot \bar{\bvec{x}} - c(\bar{\bvec{x}}) \geq \nabla_{\max} \bar{u}(\hat{\bvec{x}}') \cdot \hat{\bvec{x}}' -c(\hat{\bvec{x}}')  \geq \nabla_{\max} \bar{u}(\bvec{x}') \cdot \bvec{x}' -  c(\bvec{x}').$$
 	\end{enumerate}
 	
 	
 	
 	As a result,  the constructed $\bar{u}(\bvec{x})$  is feasible to FOP \eqref{lp:optimaldeception} because $\nabla_{max} u(\bar{\bvec{x}}) \cdot \bar{\bvec{x}} - c(\bar{\bvec{x}}) \geq \nabla_{max} u(\bvec{x}') \cdot \bvec{x}' - c(\bvec{x}')$ for any $\bvec{x}' \in X$. 
 	
 	Next, we argue that $\bar{u}(\bvec{x})$ achieves the same buyer utility, and thus must also be optimal. This is simply because the feasibility of $\bar{u}(\bvec{x})$ implies that the optimal buyer bundle will still  be $\bar{\bvec{x}}$ and payment will still be    $\bar{p} =   \nabla_{\max} u^*(\bar{\bvec{x}}) \cdot \bar{\bvec{x}}$.  As a result, buyer achieves the same utility when using $\bar{u}(\bvec{x})$ and $u^*(\bvec{x})$, yielding the optimality of $\bar{u}(\bvec{x})$.
 	
 	So far we showed that there always exists an optimal $u(\bvec{x})$ of Form \eqref{eq:optimalformat_lem}. Therefore,   to solve FOP \eqref{lp:optimaldeception}, we can without loss of generality focus on functions of the Form \eqref{eq:optimalformat_lem}, which is parameterized $p$. By plugging in Form \eqref{eq:optimalformat_lem} into FOP \eqref{lp:optimaldeception}, we obtain the following LP with variable $p \in \RR$.  
 	
 	\begin{lp}
 		\maxi{ v(\bar{\bvec{x}}) - p  }
 		\st 
 		\qcon{ p - c(\bar{\bvec{x}}) \geq \alpha \cdot  p - c(\bvec{x}') }{ \bvec{x}' =\alpha \cdot  \bar{\bvec{x}}, \alpha \in [0,1] }
 		\qcon{p - c(\bar{\bvec{x}}) \geq 0 - c(\bvec{x}' )}{\bvec{x}'  \text{ element-wise greater than }\bvec{\bar{x}} }
 		\qcon{ p - c(\bar{\bvec{x}}) \geq \min_{\{i|x'_i \leq \bar{x}_i\}} \frac{x'_i}{\bar{x}_i} \cdot  p - c(\bvec{x}' )}{\text{other } \bvec{x}' \in X}
 		\con{ p \geq 0 }
 	\end{lp}
 	We now further simplify the above LP to become LP \eqref{lp:buyer-simple-lemma}. That is, we argue that only the first constraint is needed and thus the other constraints can be omitted. When the first constraint is instantiated with $\bvec{x}' = \bvec{0}$, it implies $p - c(\bar{\bvec{x}})  \geq 0$. This  immediately implies the \textit{second} and the \emph{last} constraint.  By the proof above, we know that for any $\bvec{x}' \neq \alpha \cdot \bvec{\bar{x}}$ and $\bvec{x}'$ is not element-wise greater than $\bar{x}$ either, there must exist $\hat{\bvec{x}}'=\alpha^* \cdot \bvec{\bar{x}} \, (\alpha^*\in[0,1])$ such that $\nabla_{\max} \bar{u}(\hat{\bvec{x}}') \cdot \hat{\bvec{x}}' -c(\hat{\bvec{x}}')  \geq \nabla_{\max} \bar{u}(\bvec{x}') \cdot \bvec{x}' -  c(\bvec{x}')$. Therefore, the \textit{third} constraint is guaranteed to be satisfied as long as the first constraint is satisfied. 
 	As a result, the above LP can be further simplified to LP \eqref{lp:buyer-simple-lemma}.  
 	
 	The constraint of LP \eqref{lp:buyer-simple-lemma} guarantee that the seller will maximize revenue at bundle $\bar{x}$.  Since at $x = \bar{x}$, any $i$ will minimize the term $\frac{x_i}{\bar{x_i}}$. Lemma \ref{lem:optimalpricing} then implies the optimal prices at $\bar{x}$ can be any convex combination of pricing vectors $(0,\cdots,0, \frac{\bar{p}}{\bar{x_i}}, 0, \cdots, 0), \forall i$ . The total payment will always be $\bar{p}$ under any of these optimal pricing schemes. This completes the proof. 
 \end{proof}
 
 Theorem \eqref{thm:characterization} then follows from  Lemma \ref{lem:characterization-fix-y}.  We first observe that constraint in linear program \eqref{lp:buyer-simple-lemma} can be re-written as $p \geq \frac{c(\bar{\bvec{x}}) - c(\alpha \bar{\bvec{x}}) }{1 - \alpha}$ for $\alpha \in [0,1)$ (the constraint is trivial for $\alpha = 1$). Since the objective of LP \eqref{lp:buyer-simple-lemma}  is equivalent to minimizing $p$, we thus have the optimal $p$ equals $\max_{\alpha \in [0,1) } \big[ \frac{c(\bar{\bvec{x}}) - c(\alpha \bar{\bvec{x}}) }{1 - \alpha} \big] $. 
 
 We have now characterized the optimal imitative value function for any fixed $\bar{\bvec{x}}$.  To compute the globally optimal imitative value function $u^*$, we  only need to pick the  $\bar{\bvec{x}}$ that maximizes the buyer's surplus.   By viewing $\bar{\bvec{x}}$ as a variable $\bvec{x}$, we obtain the desired form of $\bvec{x}^*$ as in Equation \eqref{eq:optimal-format}. Finally, the buyer surplus and seller revenue follow directly from the fact that purchase happens at bundle $\bvec{x}^*$ with payment $p^*$. These conclude the proof of Theorem \ref{thm:characterization}.

\section{Omitted Proofs in Section \ref{sec:instantiation}}\label{app_sec:proof-section-initantiation}

\subsection{Proof of Theorem \ref{thm:characterization1}}
\primethmcharacterizationone*
\begin{proof}
Suppose $c(\bvec{x})$ is convex and differentiable. Fix any  $\bvec{x}=\bar{\bvec{x}}$.  Consider the function $c(\alpha   \bar{\bvec{x}})$ with variable $\alpha \in [0,1)$. This is a one-dimensional convex non-decreasing function. Due to convexity,  the supremum of  $\frac{  c( \bvec{x}) -c(\alpha \bvec{x}) }{1 - \alpha} $  over $\alpha \in [0,1)$  equals precisely the derivative of  $c(\alpha \cdot \bar{\bvec{x}})$ at $\alpha = 1$, which is  $\bvec{x} \cdot \nabla  c(\bvec{x})$. To find the $\bar{\bvec{x}}$ that maximizes the buyer's revenue, the buyer will pick $\bvec{x}^* = \arg \max_{\bvec{x} \in X} ~[v(\bvec{x}) - \bvec{x} \cdot \nabla  c(\bvec{x})]$.  Given the above characterization, the theorem conclusion follows from Theorem \ref{thm:characterization}.    
\end{proof}

\subsection{Proof of Theorem \ref{thm:characterization2}}
\primethmcharacterizationtwo*
\begin{proof}
	Suppose $c(\bvec{x})$ is concave and differentiable. Fix any  $\bvec{x}=\bar{\bvec{x}}$.  Consider the function $c(\alpha   \bar{\bvec{x}})$ with variable $\alpha \in [0,1)$. This is a one-dimensional concave non-decreasing function. Due to concavity,  the supremum of  $\frac{  c( \bvec{x}) -c(\alpha \bvec{x}) }{1 - \alpha} $  over $\alpha \in [0,1)$  is achieved at $\alpha  = 0$. The supremum thus equals precisely $\frac{  c( \bvec{x}) -c(0 \cdot  \bvec{x}) }{1 -0} = c(\bvec{x})$.  To find the $\bar{\bvec{x}}$ that maximizes the buyer's revenue, the buyer will pick $\bvec{x}^* = \arg \max_{\bvec{x} \in X} ~[v(\bvec{x}) - c(\bvec{x})]$.  Given the above characterization, the theorem conclusion follows from Theorem \ref{thm:characterization}.  Specifically, the buyer payment $p^* = c(\bvec{x}^*)$,   leading to seller revenue $0$. 
\end{proof}

\subsection{Proof of Theorem \ref{prop:hardness}}
\primethmhardness*
\begin{proof}	
	As stated in the theorem, we  consider the case where the buyer's true value function is $v(\bvec{x})  =  \sum_{i=1}^d x_i$ and the seller's production cost function $c(\bvec{x})$ is a concave function that we will construct.  Theorem \ref{thm:characterization2} shows in this case, the buyer's optimal surplus is   
	$
	\max_{\bvec{x} \in X} [ \sum_{i=1}^d x_i -    c(\bvec{x}) ]. 
	$
	
	Next, we show that this optimization problem is NP-hard to be approximated within any meaningful ratio, as described by the theorem. 	Our reduction is from the independent set problem. For any  connected graph $G = (V, E)$ with $d$ nodes, let node set $ V = [d]  = \{ 1, 2, \cdots, d\}$.   A set $I$ is an independent set of $G$ if and only if any $i, j \in I$ are not adjacent in $G$. The problem of finding the largest independent set problem is NP-hard, and cannot be approximated within  ratio $1/d^{1-\epsilon}$ for any constant $\epsilon >0$.  
	
	Given any instance graph $G = (V, E)$ of the Independent set problem, we construct the following concave production cost function: $$ c(\bvec{x}) = \sum_{i=1}^d \min(\sum_{j=1}^d a_{ji}  x_j, x_i), \text{ where } a_{ji} = 1 \text{ if }(j,i) \in E \text{ and } a_{ji} = 0 \text{ otherwise.}$$  
	Moreover, the set of feasible bundles is $X = [0,1]^d$.   Note that $c(\bvec{x})$ is a concave function because the minimum of two linear functions is concave and the sum of concave functions remains concave. Moreover, $c(\bvec{x})$ is monotone non-decreasing and $c(\bvec{0}) = 0$. So $c(\bvec{x})$ is indeed a valid cost function for our setting.   Under this  construction, the maximum possible buyer surplus is the optimal objective of the following optimization problem (OP): 
	\begin{equation}\label{eq:hardness-obj}
		\max_{\bvec{x} \in [0,1]^d} \, \, U(\bvec{x}) = \sum_{i=1}^d x_i - c(\bvec{x}) = \sum_{i=1}^d \bigg[ x_i - \min(\sum_{j=1}^d x_ja_{ji},x_i) \bigg].    
	\end{equation} 
	
	We now show via a reduction from the largest independent set problem that it  is NP-hard to approximate OP \eqref{eq:hardness-obj} to  be within ratio $1/d^{1- \epsilon}$ for any $\epsilon >0$. Let $I \subseteq V$ be the maximum independent set. We claim that the optimal objective value of the above optimization problem equals precisely $|I|$, the size of the maximum independent set. For convenience, let term  $U_i(\bvec{x}) = x_i - \min(\sum_{j=1}^d x_ja_{ji},x_i) $ and thus $U(\bvec{x}) = \sum_{i=1}^d U_i(\bvec{x})$. Note that $U_i(\bvec{x}) \leq 1$ for any $i \in V$.  
	
	First, we show that the optimal objective value of OP \eqref{eq:hardness-obj}  is at least $|I|$. To see this, consider $\bar{\bvec{x}}$ such that $\bar{x}_i = 1$ if $i \in I$ and $\bar{x}_i = 0$ if $i \not \in I$. We argue that for any $i \in I$, $U_i(\bar{\bvec{x}}) = 1$.  This is because $\bar{x}_j a_{ji} = 0$ for any $j$ --- for any $\bar{x}_j = 1$, we must have $a_{ji} = 0$ as the two nodes $i,j$ are both in the independent set and thus cannot have an edge between them (i.e., $a_{ji} = 0$). Therefore, $\sum_{j=1}^d \bar{x}_ja_{ji} = 0$ and thus $\min(\sum_{j=1}^d \bar{x}_ja_{ji},\bar{x}_i) = 0$.  As a consequence, for any $i \in I$, $U_i(\bar{\bvec{x}}) = 1$ and thus the objective of \eqref{eq:hardness-obj} at $\bar{\bvec{x}}$ is at least $|I|$. 
	
	Next, we   show the reverse direction, i.e.,  $\max_{\bvec{x} \in [0,1]^d} \, \, U(\bvec{x})$ is at most $|I|$. Note that $U(\bvec{x})$ is a convex function. So it must achieve the maximum at some vertex $\bvec{x}^*$ of the feasible region, which is a binary vector. Let $S^* \subseteq [d]$ denote the set of the indexes of non-zero values in $\bvec{x}^*$. First of all, for any $i \not \in S^*$, $U_i(\bvec{x}^*) \leq  x_i^* = 0$. Second, for any $i \in S^*$, if there exists $j \in S^*$ such that  $(i,j) \in E$ is an edge, then $x^*_j a_{ji}=1$ and thus $\min(\sum_{j=1}^d x^*_j a_{ji}, x^*_i) = 1$. This implies $U_i(\bvec{x}^*) = 0$.  Similarly, $U_j(\bvec{x}^*) = 0$ as well. Finally, for any $i \in S^*$ without any neighbor included in $S^*$, it is easy to see that $U_i(\bvec{x}^*) = 1$. To sum up, only the node $i \in S^*$ that does not have any neighbor included in $S^*$ can have $U_i(\bvec{x}^*) = 1$ whereas any other node $i$ has $U_i(\bvec{x}^*) = 0$. Therefore,  $U(\bvec{x}^*)$ is at most the size of the number of independent nodes in $S^*$, which is at most the size of the maximum independent set for $G$, as desired.
	
	So far we have shown that the optimal objective value of OP \eqref{eq:hardness-obj} equals precisely the size of the maximum independent set of $G$. However, we are not done yet  to prove the inapproximability of maximizing $U(\bvec{x})$.  This is because $U(\bvec{x})$ takes fractional variables as input. The fact that it is hard to find an independent set to approximate the size of the maximum independent set does not imply the hardness of finding a fractional variable $\bvec{x} \in X$ to approximate $\max U(\bvec{x})$. 
	
	To prove the inapproximability of the continuous OP \eqref{eq:hardness-obj},  we show that any $\alpha$-approximation to OP \eqref{eq:hardness-obj} can be efficiently turned into an $\alpha$-approximation to the largest independent set problem, using ideas from de-randomization. Since it is NP-hard to approximate the largest independent set problem to be within $1/d^{1 - \epsilon}$ for any $\epsilon > 0$, this will conclude the proof of the proposition.
	
	Specifically, let $\bar{\bvec{x}} \in [0,1]^d$ be any $\alpha$-approximation to OP \eqref{eq:hardness-obj}. We construct a \emph{random} binary vector $
	\bar{X}$ as follows: $\Pr( \bar{X}_i = 1) = \bar{x}_i$ and $\Pr( \bar{X}_i = 0) = 1- \bar{x}_i$ for each $i$ independently. By convexity of $U(\bvec{x})$, we have $\Ex  [U(\bar{X})] \geq  U(\Ex   [\bar{X}]) = U(\bar{\bvec{x}})$. In other words, if we pick the random solution $\bar{X}$, the expected objective is at least $U(\bar{\bvec{x}})$. By a standard de-randomization procedure (up to an additive $\epsilon$ difference due to Monte-Carlo sampling),\footnote{Specifically,  $\Ex_{\bar{X}   }   [U(\bar{X})] = \bar{x}_i \Ex_{\bar{X}   }   [U(\bar{X}|\bar{X}_i = 1)] + (1-\bar{x}_i) \Ex_{\bar{X}   }   [U(\bar{X}|\bar{X}_i = 0)]$ for each $i$. To de-randomize, we  simply calculate $\Ex_{\bar{X}   }   [U(\bar{X}|\bar{X}_i = 1)] $ and $\Ex_{\bar{X}   }   [U(\bar{X}|\bar{X}_i = 0)]$ through Monte-Carlo sampling, and  then pick the larger one.}  we can efficiently find a binary vector $\bar{\bvec{x}}'$ whose value is also at least $U(\bar{\bvec{x}})$. By a similar argument above, we know that all the independent nodes in  $\bar{\bvec{x}}'$ form an independent set whose size is at least $U(\bar{\bvec{x}})$, as desired.     
\end{proof}

\section{Omitted Proofs in Section \ref{sec:nonlinear-pricing}}\label{app_sec:proof-concave}
\subsection{Proof of Lemma \ref{lemma:non_linear_pricing}}\label{app_sec:proof-non-linear-pricing}
\primelemmanonlinearpricing*
\begin{proof}[Proof of Lemma \ref{lemma:non_linear_pricing}] 
	We first prove the a function of the following format, parameterized by variable $\bar{p}$, is optimal to FOP \eqref{lp:optimal_deception_concave_pricing}: 
	\begin{equation} \label{eq:optimalformat_concave_pricing}
		\bar{u}(\bvec{x}) =\bar{p}  \cdot \min \{ \frac{x_1}{\bar{x}_1}, \cdots \frac{x_d}{\bar{x}_d}, 1\} 
	\end{equation}
	To prove this, let $u^*$ be any optimal solution  to \eqref{lp:optimal_deception_concave_pricing}. We now construct another concave value function $\bar{u}(\bvec{x})$ as follows:
	\begin{equation} \label{eq:optimalformat_specific}
		\bar{u}(\bvec{x}) =u^* (\bar{\bvec{x}}) \cdot \min \{ \frac{x_1}{\bar{x}_1}, \cdots \frac{x_d}{\bar{x}_d}, 1\} 
	\end{equation} 
	Next, we first argue that the constructed $\bar{u}$
	is still feasible to FOP \eqref{lp:optimal_deception_concave_pricing}. First,   for any $\bvec{x}' =\alpha \bvec{\bar{x}}$ where $\alpha \in [0,1]$, we have  
	\begin{align}\label{eq:proof-concave-pricing}
		\bar{u}(\bar{\bvec{x}}) - c(\bar{\bvec{x}}) &= u^*(\bar{\bvec{x}}) - c(\bar{\bvec{x}}) & \mbox{(by definition of $\bar{u}$)}\\ 
		& \geq u^*(\bvec{x}')  - c(\bvec{x}') & \mbox{(by feasibility of $u^*$ ) }\\ 
		&\geq \alpha \cdot u^*(\bar{\bvec{x}}) - c(\bvec{x}')& \mbox{(By concavity of $u^*$, and $\bvec{x}'=\alpha \bvec{\bar{x}}$)}\\
		& \label{eq:proof_concave_pricing_last_step}= \bar{u}(\bvec{x}')- c(\bvec{x}'). & \mbox{(by definition of $ \bar{u}$) } 	\end{align}
	Note the last inequality is by concavity of $u^*$, we have $ (1-\alpha)\cdot u^*(\bvec{0}) + \alpha\cdot u^*(\bvec{\bar{x}}) \leq u^*(\alpha \bvec{\bar{x}})$ where we have $u^*(\bvec{0}) = 0$. Thus, we have $\alpha \cdot u^* (\bar{\bvec{x}})\leq u^*(\alpha \bvec{\bar{x}}) = u^*(\bvec{x}')$.
	
	Then we consider the case $\bvec{x}' \neq \alpha \bvec{\bar{x}}$ .  There will be two possible situations if $\bvec{x}' \neq \alpha \bar{\bvec{x}}$:
	\begin{enumerate}
		\item If $\bvec{x}' $ is element-wise greater than $\bvec{\bar{x}}$, then $ \bar{u}(\bvec{x}')  =u^*(\bar{\bvec{x}}) \cdot \min \{ \frac{x'_1}{\bar{x}_1}, \cdots \frac{x'_d}{\bar{x}_d}, 1\} = \bar{u}(\bvec{\bar{x}})$. Thus, we have $\bar{u}(\bvec{\bar{x}})- c(\bar{\bvec{x}}) \geq  \bar{u}(\bvec{x}') - c(\bvec{x}')$ by the monotonicity of $c(\bvec{x})$.
		\item If $\bvec{x}' $ is \textit{not} element-wise greater than $\bvec{\bar{x}}$, let $i^* = \argmin_{\{i|x'_i \leq \bar{x}_i\}} \frac{x'_i}{\bar{x}_i} $ and $\alpha^* = \min_{\{i|x'_i \leq \bar{x}_i\}} \frac{x'_i}{\bar{x}_i} $. Then we have  $\bar{u}(\bvec{x}') = u^*(\bar{\bvec{x}}) \cdot \min \{ \frac{x'_1}{\bar{x}_1}, \cdots \frac{x'_d}{\bar{x}_d}, 1\}= \alpha^* \, u^* (\bar{\bvec{x}}) $. In addition, denote $\hat{\bvec{x}}' = \alpha^*  \bar{\bvec{x}}$. Note that for $\hat{\bvec{x}}' $, we also have $\bar{u} (\hat{\bvec{x}}' ) = u^* (\bar{\bvec{x}}) \cdot \min \{ \frac{\hat{x}'_1}{\bar{x}_1}, \cdots \frac{\hat{x}'_d}{\bar{x}_d}, 1\}= \alpha^* \, u^* (\bar{\bvec{x}}) = \bar{u}(\bvec{x}') $. However, we have $c(\hat{\bvec{x}}') \leq c(\bvec{x}')$ because $\hat{\bvec{x}}'$ is element-wise less than or equal to $\bvec{x}'$. Thus, we have $\bar{u}(\hat{\bvec{x}}')  -c(\hat{\bvec{x}}') \geq \bar{u}(\bvec{x}')-  c(\bvec{x}') $ by monotonicity of $c(\bvec{x})$. By equations \eqref{eq:proof-concave-pricing}-\eqref{eq:proof_concave_pricing_last_step}, we have $\bar{u}(\bar{\bvec{x}}) - c(\bar{\bvec{x}})  \geq \bar{u}(\alpha \bvec{\bar{x}})  -c(\alpha \bvec{\bar{x}}), \forall \alpha \in [0,1]$. On the other hand, for any $\bvec{x}' \neq \alpha \bvec{\bar{x}} $ which is not element-wise greater than $\bvec{\bar{x}}$, there must exist a $\hat{\bvec{x}}' = \alpha^* \bar{\bvec{x}}$ where $\alpha^*= \min_{\{i|x'_i \leq \bar{x}_i\}} \frac{x'_i}{\bar{x}_i} \in [0,1]$ such that $\bar{u}(\bvec{x}')-  c(\bvec{x}') \leq \bar{u}(\hat{\bvec{x}}')  -c(\hat{\bvec{x}}')  \leq \bar{u}(\bar{\bvec{x}}) - c(\bar{\bvec{x}}) $.
	\end{enumerate}

	As a result,  the constructed $\bar{u}(\bvec{x})$  is feasible to FOP \eqref{lp:optimal_deception_concave_pricing} because $\bar{u}(\bar{\bvec{x}}) - c(\bar{\bvec{x}})  \geq \bar{u}(\bvec{x}')-  c(\bvec{x}') $ for any $\bvec{x}'\in X$ .
	
	Next, we argue that $\bar{u}$ achieves the same buyer utility, and thus must also be optimal. This is because: (1) the feasiblity of $\bar{u}$ implies that the optimal price will be $\bar{u}(\bar{\bvec{x}}) = u^*(\bar{\bvec{x}})$; (2) the optimal buyer amount will then be $\bar{\bvec{x}}$ by breaking ties rule. As a result, buyer achieves the same utility when using $\bar{u}$ and $u^*$, yielding the optimality of $\bar{u}$.
	
	So far we showed that there always exists an optimal $\bar{u}$ of Form \eqref{eq:optimalformat_concave_pricing}, which is parameterized $\bar{p} \in \mathbb{R}_{\geq 0}$. We then return to the situation of linear pricing since a linear pricing scheme with unit price $ (\lambda_1 \cdot \frac{\bar{p}}{\bar{x}_1}, \lambda_2 \cdot \frac{\bar{p}}{\bar{x}_2}, \cdots, \lambda_d \cdot \frac{\bar{p}}{\bar{x}_d})$ where $\sum_i \lambda_i = 1$ is  optimal for such a  $\bar{u}$ among all concave pricing schemes. The characterization then follows from Theorem \ref{thm:characterization}. 
\end{proof}

\subsection{Proof of Proposition \ref{thm:concave_sub_class}} \label{app_sec:proof_concave_sub_class}
\primepropsubclass*
\begin{proof}
	We first analyze the situation of linear pricing. By the characterization in Theorem \ref{thm:characterization1}, we know that the optimal purchase bundle satisfies $x^* = \argmax_{\mathbf{x}\in X}[v(\mathbf{x}) - \mathbf{x} \cdot \nabla c(\mathbf{x})]$. Given $v(x)$ and $c(x)$ in Example \ref{ex:hurt}, this solves for $x^* = 0.81$ in the linear pricing setting. Furthermore, this gives $p^* = x^* \cdot \nabla c(x^*) = 1.3122$ and optimal imitative value function $u^* = 1.3122 \cdot \min\{x/0.81, 1\}$. The seller revenue in this case is $p^* - c(x^*) = 0.6561$.
	
	Next, we show that the seller's revenue will strictly decrease at equilibrium when using pricing class $\P = \P_L \cup \{ \tilde{p} \}$ that augments linear pricing class $\P_L$  with the following additional choice of a concave pricing function
	\[
	\tilde{p} (x) = \min(\sqrt{x}, \frac{5}{9}x + \frac{9}{20} - \epsilon)
	\]
	where  $\epsilon = 0.05$  is a   small constant. 
	
	Consider the imitative value function $u^*(x) = \sqrt{x}$ and a particular linear pricing response $p > 0$. In this case, the buyer will purchase an amount $x' = \argmax_{x\in X}[u^*(x) - p\cdot x]$, implying $\frac{1}{2\sqrt{x}} = \frac{du^*(x')}{dx} = p$. Solving for $x'$ gives $x' = \frac{1}{4p^2}$. Now, we solve for the linear pricing response $p_L$ that maximizes the seller revenue.
	
	\begin{align*}
		p_L &= \argmax_{p \in L}[p \cdot x' - c(x')] \\
		&= \argmax_{p \in L}[p \cdot \frac{1}{4p^2} - \frac{1}{(4p^2)^2}] \\
		&= \argmax_{p \in L}[\frac{1}{4p} - \frac{1}{16p^4}] \\
		&= 1
	\end{align*}
	
	Thus, the optimal linear pricing response to $u^*(x) = \sqrt{x}$ is $p_L = 1$. The buyer will purchase $x' = 0.25$, giving the seller a revenue of $1 \cdot 0.25 - 0.25^2 = 0.1875$. Finally, consider the pricing function $\tilde{p}(x)$ for the buyer's imitative value function $u^*(x) = \sqrt{x}$. We observe that $\frac{d}{dx}(\frac{5}{9}x + \frac{9}{20} - \epsilon) = \frac{5}{9} = \frac{1}{2\sqrt{0.81}} = \frac{d}{dx}(u^*(0.81))$, meaning $\argmax_{x \in X}[u^*(x) - p^*(x)] = 0.81$. Thus, the buyer can purchase $x = 0.81$ and the seller will get revenue $0.2439 - \epsilon = 0.1939$. This is strictly larger than the seller's revenue $0.1875$ from the optimal linear pricing scheme. Thus the seller's optimal pricing scheme from $\P$ is  $\tilde{p}(x)$. 
	
	Finally, note that $\argmax_{x \in X}[v(x) - \tilde{p}(x)] = 0.81$, meaning the optimal bundle for the buyer to purchase when the seller responds with $\tilde{p}$ is $x = 0.81$. In this case the buyer surplus is $v(0.81) - \tilde{p}(0.81) = 8.1 - 0.9 + \epsilon > 8.1 - 1.3122 = v(0.81) - p_L \cdot 0.81$, meaning the optimal true buyer surplus given pricing function $\tilde{p}(x)$ is greater than the optimal true buyer surplus given any linear pricing response. Thus, $u^*(x)$ is an optimal imitative function under pricing class $\P$ and the seller revenue of $0.2439 - \epsilon = 0.1939$ is strictly lower than the seller revenue of $0.6561$ under linear pricing class $\P_L$.
\end{proof}

%% file: main.bbl
\begin{thebibliography}{10}

\bibitem{akyol2016price}
E.~Akyol, C.~Langbort, and T.~Basar.
\newblock Price of transparency in strategic machine learning.
\newblock {\em arXiv}, pages arXiv--1610, 2016.

\bibitem{allen1967macro}
R.~Allen.
\newblock Macro-economic theory: a mathematical treatment.
\newblock 1967.

\bibitem{amin2015online}
K.~Amin, R.~Cummings, L.~Dworkin, M.~Kearns, and A.~Roth.
\newblock Online learning and profit maximization from revealed preferences.
\newblock In {\em Twenty-Ninth AAAI Conference on Artificial Intelligence},
  2015.

\bibitem{amin2013learning}
K.~Amin, A.~Rostamizadeh, and U.~Syed.
\newblock Learning prices for repeated auctions with strategic buyers.
\newblock In {\em Advances in Neural Information Processing Systems}, pages
  1169--1177, 2013.

\bibitem{amin2014repeated}
K.~Amin, A.~Rostamizadeh, and U.~Syed.
\newblock Repeated contextual auctions with strategic buyers.
\newblock In {\em Advances in Neural Information Processing Systems}, pages
  622--630, 2014.

\bibitem{balcan2014learning}
M.-F. Balcan, A.~Daniely, R.~Mehta, R.~Urner, and V.~V. Vazirani.
\newblock Learning economic parameters from revealed preferences.
\newblock In {\em International Conference on Web and Internet Economics},
  pages 338--353. Springer, 2014.

\bibitem{beato1985marginal}
P.~Beato and A.~Mas-Colell.
\newblock On marginal cost pricing with given tax-subsidy rules.
\newblock {\em Journal of Economic Theory}, 37(2):356--365, 1985.

\bibitem{beigman2006learning}
E.~Beigman and R.~Vohra.
\newblock Learning from revealed preference.
\newblock In {\em Proceedings of the 7th ACM Conference on Electronic
  Commerce}, pages 36--42, 2006.

\bibitem{ben1990computational}
O.~Ben-Ayed and C.~E. Blair.
\newblock Computational difficulties of bilevel linear programming.
\newblock {\em Operations Research}, 38(3):556--560, 1990.

\bibitem{birmpas2020optimally}
G.~Birmpas, J.~Gan, A.~Hollender, F.~Marmolejo, N.~Rajgopal, and A.~Voudouris.
\newblock Optimally deceiving a learning leader in stackelberg games.
\newblock {\em Advances in Neural Information Processing Systems}, 33, 2020.

\bibitem{bruckner2011stackelberg}
M.~Br{\"u}ckner and T.~Scheffer.
\newblock Stackelberg games for adversarial prediction problems.
\newblock In {\em Proceedings of the 17th ACM SIGKDD international conference
  on Knowledge discovery and data mining}, pages 547--555, 2011.

\bibitem{chen2020learning}
Y.~Chen, Y.~Liu, and C.~Podimata.
\newblock Learning strategy-aware linear classifiers.
\newblock {\em Advances in Neural Information Processing Systems}, 33, 2020.

\bibitem{chen:ec18}
Y.~Chen, C.~Podimata, A.~D. Procaccia, and N.~Shah.
\newblock Strategyproof linear regression in high dimensions.
\newblock In {\em Proceedings of the 2018 ACM Conference on Economics and
  Computation}, EC ’18, page 9–26, New York, NY, USA, 2018. Association for
  Computing Machinery.

\bibitem{conitzer2006computing}
V.~Conitzer and T.~Sandholm.
\newblock Computing the optimal strategy to commit to.
\newblock In {\em Proceedings of the 7th ACM conference on Electronic
  commerce}, pages 82--90, 2006.

\bibitem{dekel2010incentive}
O.~Dekel, F.~Fischer, and A.~D. Procaccia.
\newblock Incentive compatible regression learning.
\newblock {\em Journal of Computer and System Sciences}, 76(8):759--777, 2010.

\bibitem{den2015dynamic}
A.~V. den Boer.
\newblock Dynamic pricing and learning: historical origins, current research,
  and new directions.
\newblock {\em Surveys in operations research and management science},
  20(1):1--18, 2015.

\bibitem{dong:ec18}
J.~Dong, A.~Roth, Z.~Schutzman, B.~Waggoner, and Z.~S. Wu.
\newblock Strategic classification from revealed preferences.
\newblock In {\em Proceedings of the 2018 ACM Conference on Economics and
  Computation}, EC ’18, page 55–70, New York, NY, USA, 2018. Association
  for Computing Machinery.

\bibitem{gan2019manipulating}
J.~Gan, Q.~Guo, L.~Tran-Thanh, B.~An, and M.~Wooldridge.
\newblock Manipulating a learning defender and ways to counteract.
\newblock In {\em Advances in Neural Information Processing Systems}, pages
  8272--8281, 2019.

\bibitem{gan2019imitative}
J.~Gan, H.~Xu, Q.~Guo, L.~Tran-Thanh, Z.~Rabinovich, and M.~Wooldridge.
\newblock Imitative follower deception in stackelberg games.
\newblock In {\em Proceedings of the 2019 ACM Conference on Economics and
  Computation}, pages 639--657, 2019.

\bibitem{hardt2016stratclass}
M.~Hardt, N.~Megiddo, C.~Papadimitriou, and M.~Wootters.
\newblock Strategic classification.
\newblock In {\em Proceedings of the 2016 ACM Conference on Innovations in
  Theoretical Computer Science}, ITCS ’16, page 111–122, New York, NY, USA,
  2016. Association for Computing Machinery.

\bibitem{hu:fat19}
L.~Hu, N.~Immorlica, and J.~W. Vaughan.
\newblock The disparate effects of strategic manipulation.
\newblock In {\em Proceedings of the Conference on Fairness, Accountability,
  and Transparency}, FAT* ’19, page 259–268, New York, NY, USA, 2019.
  Association for Computing Machinery.

\bibitem{korula2015optimizing}
N.~Korula, V.~Mirrokni, and H.~Nazerzadeh.
\newblock Optimizing display advertising markets: Challenges and directions.
\newblock {\em IEEE Internet Computing}, 20(1):28--35, 2015.

\bibitem{li2014consumers}
J.~Li, N.~Granados, and S.~Netessine.
\newblock Are consumers strategic? structural estimation from the air-travel
  industry.
\newblock {\em Management Science}, 60(9):2114--2137, 2014.

\bibitem{mas1995microeconomic}
A.~Mas-Colell, M.~D. Whinston, J.~R. Green, et~al.
\newblock {\em Microeconomic theory}, volume~1.
\newblock Oxford university press New York, 1995.

\bibitem{milli:fat19}
S.~Milli, J.~Miller, A.~D. Dragan, and M.~Hardt.
\newblock The social cost of strategic classification.
\newblock In {\em Proceedings of the Conference on Fairness, Accountability,
  and Transparency}, FAT* ’19, page 230–239, New York, NY, USA, 2019.
  Association for Computing Machinery.

\bibitem{mohri2014optimal}
M.~Mohri and A.~Munoz.
\newblock Optimal regret minimization in posted-price auctions with strategic
  buyers.
\newblock In {\em Advances in Neural Information Processing Systems}, pages
  1871--1879, 2014.

\bibitem{mohri2015revenue}
M.~Mohri and A.~Munoz.
\newblock Revenue optimization against strategic buyers.
\newblock In {\em Advances in Neural Information Processing Systems}, pages
  2530--2538, 2015.

\bibitem{nedelec2020robust}
T.~Nedelec, C.~Calauzenes, V.~Perchet, and N.~El~Karoui.
\newblock Robust stackelberg buyers in repeated auctions.
\newblock In {\em International Conference on Artificial Intelligence and
  Statistics}, pages 1342--1351. PMLR, 2020.

\bibitem{nguyen2020decoding}
T.~H. Nguyen, N.~Vu, A.~Yadav, and U.~Nguyen.
\newblock Decoding the imitation security game: Handling attacker imitative
  behavior deception.
\newblock In {\em 24th European Conference on Artificial Intelligence, ECAI
  2020, including 10th Conference on Prestigious Applications of Artificial
  Intelligence, PAIS 2020}, pages 179--186. IOS Press BV, 2020.

\bibitem{nguyen2019imitative}
T.~H. Nguyen and H.~Xu.
\newblock Imitative attacker deception in stackelberg security games.
\newblock In {\em Proceedings of the 28th International Joint Conference on
  Artificial Intelligence}, pages 528--534. AAAI Press, 2019.

\bibitem{Perote2004StrategyproofEF}
J.~Perote and J.~Perote-Pe{\~n}a.
\newblock Strategy-proof estimators for simple regression.
\newblock {\em Math. Soc. Sci.}, 47:153--176, 2004.

\bibitem{roth2020multidimensional}
A.~Roth, A.~Slivkins, J.~Ullman, and Z.~S. Wu.
\newblock Multidimensional dynamic pricing for welfare maximization.
\newblock {\em ACM Transactions on Economics and Computation (TEAC)},
  8(1):1--35, 2020.

\bibitem{roth2016watch}
A.~Roth, J.~Ullman, and Z.~S. Wu.
\newblock Watch and learn: Optimizing from revealed preferences feedback.
\newblock In {\em Proceedings of the forty-eighth annual ACM symposium on
  Theory of Computing}, pages 949--962, 2016.

\bibitem{samuelson1998microeconomics}
P.~Samuelson and W.~Nordhaus.
\newblock {\em Microeconomics}.
\newblock Irwin/McGraw-Hill, 1998.

\bibitem{shephard1974law}
R.~W. Shephard and R.~F{\"a}re.
\newblock The law of diminishing returns.
\newblock {\em Zeitschrift f{\"u}r National{\"o}konomie}, 34(1-2):69--90, 1974.

\bibitem{tang2018price}
P.~Tang and Y.~Zeng.
\newblock The price of prior dependence in auctions.
\newblock In {\em Proceedings of the 2018 ACM Conference on Economics and
  Computation}, pages 485--502, 2018.

\bibitem{turvey1969marginal}
R.~Turvey.
\newblock Marginal cost.
\newblock {\em The Economic Journal}, 79(314):282--299, 1969.

\bibitem{vanunts2019optimal}
A.~Vanunts and A.~Drutsa.
\newblock Optimal pricing in repeated posted-price auctions with different
  patience of the seller and the buyer.
\newblock In {\em Advances in Neural Information Processing Systems}, pages
  939--951, 2019.

\bibitem{zadimoghaddam2012efficiently}
M.~Zadimoghaddam and A.~Roth.
\newblock Efficiently learning from revealed preference.
\newblock In {\em International Workshop on Internet and Network Economics},
  pages 114--127. Springer, 2012.

\bibitem{zhang2019distinguishing}
H.~Zhang, Y.~Cheng, and V.~Conitzer.
\newblock Distinguishing distributions when samples are strategically
  transformed.
\newblock In {\em Advances in Neural Information Processing Systems}, pages
  3193--3201, 2019.

\bibitem{zhang2019samples}
H.~Zhang, Y.~Cheng, and V.~Conitzer.
\newblock When samples are strategically selected.
\newblock In {\em International Conference on Machine Learning}, pages
  7345--7353, 2019.

\bibitem{pmlr-v119-zhang20f}
H.~Zhang and V.~Conitzer.
\newblock Learning the valuations of a $k$-demand agent.
\newblock In H.~D. III and A.~Singh, editors, {\em Proceedings of the 37th
  International Conference on Machine Learning}, volume 119 of {\em Proceedings
  of Machine Learning Research}, pages 11066--11075. PMLR, 13--18 Jul 2020.

\bibitem{zhangincentive}
H.~Zhang and V.~Conitzer.
\newblock Incentive-aware pac learning.
\newblock {\em AAAI 2021}, 2021.

\end{thebibliography}
